\documentclass[review,12pt]{elsarticle}
\usepackage[margin=1in,footskip=0.25in]{geometry}
\usepackage{times}
\usepackage{amssymb}
\usepackage{graphicx}
\usepackage{amsmath,amsthm,mathtools}
\usepackage[version=4]{mhchem}
\usepackage{siunitx}
\usepackage{float,subcaption}
\usepackage{longtable,tabularx}
\usepackage{multirow}
\usepackage{xcolor}
\usepackage{hyperref}
\hypersetup{colorlinks, breaklinks, citecolor=blue, linkcolor=blue, urlcolor=blue}
\usepackage[nameinlink,noabbrev,sort&compress]{cleveref}
\Crefname{equation}{Eq.}{Eqs.}
\Crefname{figure}{Fig.}{Figs.}
\crefname{figure}{Figure}{Figures}

\newtheorem{lemma}{Lemma}

\newtheorem{problem}{Problem}
\newtheorem{remark}{Remark}
\newtheorem{assumption}{Assumption}

\newcommand{\sign}{\mathrm{sign}}
\newcommand{\diag}{\mathrm{diag}}

\biboptions{numbers,sort&compress}
\makeatletter
\def\ps@pprintTitle{%
	\let\@oddhead\@empty
	\let\@evenhead\@empty
	\let\@oddfoot\@empty
	\let\@evenfoot\@empty
}
\makeatother
\begin{document}
	
	\begin{frontmatter}
		\title{Integrated Guidance and Control for Path-Following with Bounded Inputs} 
		\author[rmv]{Ram Milan Kumar Verma\corref{cor1}} 
		\cortext[cor1]{Corresponding author}
		\ead{rmverma@aero.iitb.ac.in}
		\author[rmv]{Shashi Ranjan Kumar}
		\ead{srk@aero.iitb.ac.in}
		\author[rmv]{Hemendra Arya}
		\ead{arya@aero.iitb.ac.in}
		\affiliation[rmv]{organization={Indian Institute of Technology Bombay},
			addressline={Powai}, 
			city={Mumbai},
			postcode={400076}, 
			state={Maharashtra},
			country={India}}
		\begin{abstract}
			Precise motion control of underactuated surface vessels is a crucial task in various maritime applications. In this work, we develop a nonlinear motion control strategy for surface vessels inspired by the pursuit guidance philosophy. Any sufficiently smooth path can be seen as a continuum of virtual targets moving along a specified path, which the pursuer is trying to catch. Contrary to the traditional path-following methods, this work develops an integrated guidance and control approach capable of following any smooth path (unlike the ones composed of a finite number of straight lines and circles). The approach relies on steering the vehicle such that its velocity vector aligns with the line-of-sight (the line joining the moving virtual target and the surface vessel), resulting in a tail-chase scenario. This leads to a path-following behavior. This integrated approach also overcomes the disadvantages inherent in the traditional two-loop-based approaches.  Additionally, the proposed work takes into account the asymmetric actuator constraints in the design, which makes the design close to realistic scenarios. Furthermore, the control law has been derived within a nonlinear framework using sliding mode, thus remains applicable for wider envelope. The stability of the proposed control strategy is formally proven. Numerical simulations for various specified paths validate the controller's accurate path-following performance.
		\end{abstract}
		
	\end{frontmatter}
	\section{Introduction}
	Motivated by the versatility of uncrewed surface vessels (USVs) in addressing challenging real-world problems in the maritime domain, there is a rapidly growing need for advanced guidance strategies. USVs are widely used in various commercial, industrial, defense, and research applications. It encompasses various applications, including tourism, oceanography, climate monitoring, naval mine detection, offshore infrastructure monitoring, cable laying on the sea floor, aquaculture monitoring, and floating waste collection near harbors, among others.
	USVs with autonomous capabilities can perform such tasks and eliminate human danger in harsh marine environments.
	
	Most of the problems' solutions narrow down to path-following or trajectory-tracking requirements to meet the specific mission objectives. However, the path or the trajectory can be of any arbitrary shape. Therefore, in this work, we propose an integrated guidance approach for a USV that allows the vehicle to follow an arbitrary smooth path. 
	
	Many of the existing works on the path-following strategies had been employed for different vehicles such as aerial, ground, and marine crafts. The major differences lie in the nature of the forces acting on the system and in the degrees of freedom of motion. One major distinction between marine and aerial vehicles is the effect of added mass, which is negligible in aerial vehicles because air density is much lower than water density. Also, aerial applications require a three-dimensional strategy, whereas ground and USV applications require a two-dimensional strategy.
	In \cite{2013_ecc_pbsujit_ecc,10.1109/MCS.2013.2287568}, various path-following strategies for unmanned aerial vehicles (UAVs) had been presented.
	The existing approaches to path-following include the carrot chasing algorithm, vector field-based path-following \cite{nelson2007vector,10.1016/j.oceaneng.2015.12.037}, and pure pursuit and line-of-sight based
	path-following \cite{breivik2007applying,fossen2003line,1582226, 10.1109/TCST.2014.2338354, 10.1109/TCST.2023.3259819,10.1109/JOE.2024.3483328}. Model predictive control (MPC) based strategies were also applied to the path-following in \cite{10.1016/j.oceaneng.2009.10.004}. In \cite{10.1109/JOE.2021.3059210}, the authors proposed a LOS-based law with a time-varying look-ahead distance. In these approaches, a desired heading was first found out, and then suitable design methods were used to nullify the heading error. The most common of these methods was different forms of proportional-integral-derivative controllers. However, in all cases, the paths were treated as compositions of waypoints, lines, and circles. Note that, in general, a path can be broken down into the form of lines and arcs of circles of different radii, but in order to get a sufficiently smooth path, there may be too many waypoints and circular arcs. 
	
	The authors in \cite{kumar2024robust} develop a path-following strategy for UAVs using pursuit guidance. 
	However, it only considered point-mass kinematics and used both linear and angular speeds as control inputs. 
	In \cite{kumar2025robust}, the authors presented a robust trajectory tracking algorithm by using a two-loop-based method to control the motion of a quadrotor. 
	In \cite{10.1109/TCST.2023.3259819}, the authors design an adaptive guidance law by first designing a desired heading based on the cross-track error and then a control law, which makes it a two-loop-based strategy. 
	However, in our work, we consider the nonlinear dynamics of the USV and develop an integrated approach that directly designs the control inputs (that is, surge thrust and yaw torque) for the USV.
	This integrated guidance and control approach offers a significant advantage over conventional two-loop-based guidance and control, as evident from \cite{shima2006igc}.  
	The two-loop-based control is based on the premise that the inner loop is faster than the outer loop \cite{beard2012small}. 
	However, this integrated approach inherently takes into account the system's dynamics, thereby improving the performance.
	
	Actuator saturation is another crucial problem that is often overlooked while designing control strategies. All practical actuators can produce only a bounded thrust/torque. The works in \cite{10.1016/j.oceaneng.2024.116956,yu2019elos,10.1016/j.oceaneng.2024.118217} propose LOS-based guidance laws that account for input saturation. Numerous other works that account for input saturation while designing path-following strategies are \cite{10.1016/j.oceaneng.2022.112327, ZHENG2016158,10.1109/JOE.2021.3059210}. The approaches presented in these works approximate the saturation function using smooth hyperbolic tangent or sigmoid functions. Moreover, they only consider symmetric bounds on the actuators. However, in practice, the actuators can have different capabilities in forward and reverse directions, due to cost-cutting measures or asymmetric wear and tear resulting from usage.
	
	In this work, we present a two-dimensional path-following integrated guidance strategy for USVs, that also accounts for the nonlinear dynamics of the USV and enforces asymmetric actuator constraints. We draw inspiration from the guidance strategies \cite{zarchan2012tactical} in the interceptor guidance literature. However, there is a significant deviation, as the interceptor is assumed to be moving with a constant speed, whereas in our case, the USV's speed is considered to be time-varying. In light of the above discussions, the main contributions of this paper are highlighted next:
	\begin{itemize}
		\item The proposed work presents path-following strategies for any sufficiently smooth path, unlike in \cite{10.1109/TCST.2014.2338354}, which breaks the path into lines and circles and then designs a strategy to follow the path.
		\item First, a strategy is proposed without consideration of actuator constraints in the design, with physical actuator bounds subsequently enforced by saturating the control demands at their maximum and minimum bounds during its implementation.
		\item Second, this work also incorporates a smooth input saturation model in the design to account for asymmetric actuator constraints. This ensures that the control input demand never exceeds its bounds, resulting in a smooth control profile that extends the actuator's life while guaranteeing the system stability.
		\item Instead of using a point mass kinematic model (as in \cite{kumar2024robust}), we consider the complete nonlinear dynamics of the USV as well as the nonlinear engagement kinematics to derive the control law, which remains valid for a wide range of operating conditions. This proposed scheme is designed for an underactuated vehicle, a common scenario in many practical applications.
		\item The proposed scheme employs an integrated guidance and control (IGC) approach, which offers significant advantages over traditional two-loop guidance strategies \cite{shima2006igc}. Additionally, we utilize sliding mode control, which inherently provides robustness against matched uncertainties.
		\item Though the philosophy of the proposed scheme is inspired by interceptor guidance literature, this work deviates significantly as it accounts for time-varying speeds of the USV and is capable of following the path with time-varying speeds over the path. 
	\end{itemize}
	Apart from rigorously proving the system's stability under the proposed scheme, it has also been validated through numerical simulations for various types of paths.
	
	The rest of the paper is organized as follows: \Cref{sec:problem} presents the USV dynamics and the engagement kinematics of the USV with a virtual reference point moving along the path, followed by the problem formulation. The derivation of the necessary controller for the vehicle, along with its stability proof, will be presented in \Cref{sec:main_results}, after briefly discussing the necessary preliminaries. \Cref{sec:simulation} demonstrates the simulation results validating the efficacy of the proposed control design. \Cref{sec:conclusion} concludes the work along with indications of some possible future research directions.
	
	\section{Problem Formulation}\label{sec:problem}
	In this section, the reference frames are defined first, and then the equations of motion are described accordingly. Following this, the planar engagement scenario is presented with respect to the current USV position and the virtual reference point moving along the desired path. The engagement kinematics is described in terms of distance from the virtual reference point on the path, $\mathcal{P}$, and the orientation of the line joining them. Towards the end of the section, we formally describe the problem addressed in this paper.
	
	\subsection{USV dynamics}
	This subsection outlines the dynamics of the USV. Typically, the dynamics are described by six degrees of freedom (DOF) ordinary differential equations of motion. These DOFs are referred to as surge, sway, and heave for the positions, and roll, pitch, and yaw for the spatial orientation of the vehicle. Here, we assume that the motion of the USV is constrained to a plane, that is, the vessel can only translate in the $x-y$ plane and rotate about the axis perpendicular to the plane called the yaw. This assumption is reasonable because the USVs are designed to be metacentrically stable with small amplitudes of roll, pitch, and their rates. Hence, the roll and pitch dynamics can be ignored. Similarly, the heave-direction motion can also be neglected, as the USV is floating on the water surface with zero mean. Hence, USV dynamics can be written in the horizontal plane as a three-DOF model. 
	\begin{figure}[ht]
		\centering
		\includegraphics[width=0.5\linewidth]{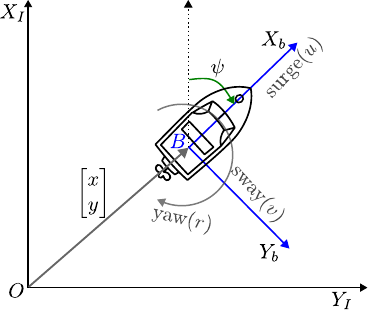}
		\caption{Schematic representation of the surface vessel's position, orientation, and body-frame velocities.}
		\label{fig:ref_frame}
	\end{figure}
	
	To describe the position and orientation of the USV, we consider two reference frames as shown in \Cref{fig:ref_frame}. One is Earth-fixed frame $(OX_IY_I)$, and the other is body-fixed frame $(BX_{b}Y_{b})$. The origin of the Earth-fixed frame is at $O$, while the origin of the body-fixed frame is at $B$. The axes of the Earth-fixed frame, $OX_I$ and $OY_I$, are directed towards the North and the East, respectively. The body-fixed frame axes $BX_b$ point from the aft to the fore, and $BY_b$ is towards the starboard side. Since the $x-y$ plane of both the frames are the same, the two frames are related by a rotation matrix, $\pmb{J}$, given by
	\begin{align}
		\pmb{J}(\psi) = \begin{bmatrix}
			\cos\psi & -\sin \psi &0\\
			\sin\psi & \cos\psi & 0 \\
			0 & 0 & 1
		\end{bmatrix},
	\end{align}
	where $\psi$ is the rotation angle relating the two frames. The vector $\pmb{\eta} = [x~~y~~\psi]^\top$ describes the position $(x,~y)$ and the heading angle $(\psi)$ of the surface vessel in the Earth-fixed frame. The vector $\pmb{\nu} = [u ~~ v ~~ r]^\top$ denotes the velocity vector components of the surface vessel in the body-fixed frame. The variable $u$ is the linear velocity in the surge, $v$ is the linear velocity in sway, and $r$ is the yaw rate of the surface vessel.  Now, the three-DOF non-linear equation of motion can be expressed, \cite{yang2014_dyamics, fossen2021handbook}, as
	\begin{align}\label{eqn: kinematics}
		\dot{\pmb{\eta}}  &= \pmb{J}(\psi) \pmb{\nu},
		\\
		\label{eqn: dynamics}
		\dot{\pmb{\nu}} &= \pmb{M}^{-1}\left[\pmb{\tau} - \pmb{C}(\pmb{\nu})\pmb{\nu} - \pmb{D}(\pmb{\nu})\pmb{\nu}  + \pmb{b}\right],
	\end{align}
	where $\pmb{\tau} = \begin{bmatrix}
		\tau_1& \tau_2& \tau_3    
	\end{bmatrix}^\top=\begin{bmatrix}
		\tau_u& \tau_v& \tau_r  
	\end{bmatrix}^\top$ is the control input in the body-fixed frame, $\pmb{M}\in \mathbb{R}^{3\times3}$ is a positive definite mass inertia matrix, $\pmb{C}\in \mathbb{R}^{3\times3}$ is the Coriolis-centripetal matrix, $\pmb{D}\in \mathbb{R}^{3\times3}$ is the damping matrix, and $\pmb{b}\in \mathbb{R}^{3}$ is the disturbance vector. The expressions for these matrices can be written as 
	\begin{align}\label{eqn: MCD}
		\nonumber  \pmb{M} &= 
		\begin{bmatrix}
			m_{11} &0 &0\\
			0 &m_{22} &m_{23} \\
			0  &m_{32} &m_{33}
		\end{bmatrix}, 
		\quad 
		\pmb{C}(\pmb{\nu}) = 
		\begin{bmatrix}
			0 &0 &-m_{22}v-m_{23}r\\
			0 &0 &m_{11}u \\
			m_{22}v+m_{23}r  &-m_{11}u &0
		\end{bmatrix},\\
		\pmb{D}(\pmb{\nu}) &= 
		\begin{bmatrix}
			d_{11}(u) &0 &0\\
			0 &d_{22}(v,r) &d_{23}(v,r) \\
			0  &d_{32}(v,r) &d_{33}(v,r)
		\end{bmatrix}. 
	\end{align}
	The expressions for the matrix elements used in \Cref{eqn: MCD} are given by
	\begin{subequations}\label{eqn: dyn_elements}
		\begin{align}
			m_{11} &= \text{m} - X_{\dot{u}},
			m_{22} = \text{m} -  Y_{\dot{v}},\\
			m_{23} &= \text{m}x_g -  Y_{\dot{r}},
			m_{32} = \text{m}x_g -  N_{\dot{v}},\\
			m_{33} &= I_z -  N_{\dot{r}},
			d_{11}(u) = -X_u - X_{|u|u}|u|,\\
			d_{22}(v,r) &= -Y_v - Y_{|v|v}|v| - Y_{|r|v}|r|,\\
			d_{23}(v,r) &= -Y_r - Y_{|v|r}|v| - Y_{|r|r}|r|,\\
			d_{32}(v,r) &= -N_v -N_{|v|v}|v| - N_{|r|v}|r|,\\
			d_{33}(v,r) &= -N_r - N_{|v|r}|v| - N_{|r|r}|r|,
		\end{align}  
	\end{subequations}
	where $\text{m}$ is the mass of the surface vessel, $x_g$ is the distance of the geometric center of the surface vessel from the center of gravity (CG), and $I_z$  is the yaw moment of inertia of the surface vessel. It is to be noted that matrices $\pmb{M}$ and $\pmb{C}$ contain the components also related to the added mass. The remaining symbols are related to hydrodynamic derivatives and have the usual meanings, as per the SNAME convention \cite{SNAME1950nomenclature}. 
	
	The kinetics of the USV in \Cref{eqn: dynamics} can be expanded with the help of \Cref{eqn: MCD,eqn: dyn_elements} as
	\begin{subequations}
		\begin{align}
			\dot{u} &= \dfrac{m_{22}}{m_{11}}v r + \dfrac{m_{23}}{m_{11}}r^2 - \dfrac{d_{11}}{m_{11}}u + \dfrac{\tau_u}{m_{11}},\label{eqn: u__dot}\\
			\dot{v} &= -\dfrac{m_{23}}{m_{22}}\dot{r} -\dfrac{m_{11}}{m_{22}}ur -\dfrac{d_{22}}{m_{22}}v -\dfrac{d_{23}}{m_{22}}r, \label{eqn:v__dot}\\
			\dot{r} &= -\dfrac{m_{32}}{m_{33}}\dot{v} + \dfrac{m_{11} - m_{22}}{m_{33}}vu -\dfrac{m_{23}}{m_{33}}ur -\dfrac{d_{32}}{m_{33}}v -\dfrac{d_{33}}{m_{33}}r + \dfrac{\tau_r}{m_{33}}  \label{eqn:r__dot}.
		\end{align}
	\end{subequations}
	Note that since the USV is assumed to be underactuated, the control input in the sway direction, $\tau_v$, is substituted with zero as there is no independent control thruster available in that particular channel. Hence, the control input vector is represented as $\pmb{\tau} = [\tau_u~~0~~\tau_r]^\top$. However, while deriving the control inputs, we also refer to $\pmb{\tau}$ as $[\tau_u~~\tau_r]^\top$.
	
	\subsection{Engagement kinematics}
	\begin{figure}[htpb]
		\centering
		\includegraphics[width=.6\linewidth]{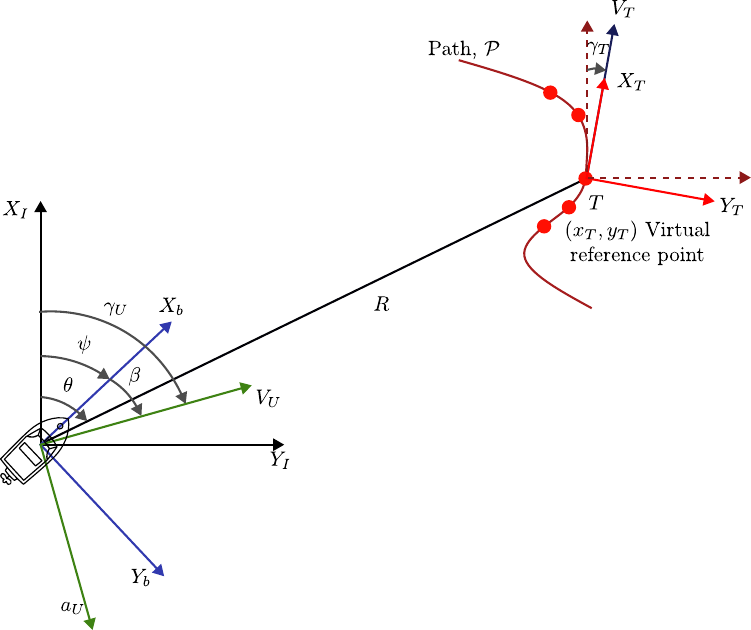}
		\caption{Kinematic engagement of USV with virtual reference point, $T$, on the path, $\mathcal{P}$.}
		\label{fig:engagement}
	\end{figure}
	Now, to pose the path-following problem similar to the interceptor guidance problem, consider the engagement of the USV as shown in \Cref{fig:engagement}. The dynamics for the relative motion of the USV with respect to the desired virtual reference point can be described in terms of the rate of change of range to go, line-of-sight (LOS) rotation rate, and the rotation rate of the USV's velocity vector, which can be written as
	\begin{align}
		\dot{R} &= V_T\cos\theta_T - V_U \cos \theta_U \label{eqn:r_dot},\\
		R\dot{\theta} &= V_T\sin\theta_T - V_U \sin \theta_U \label{eqn:r_theta_dot}, \\
		\dot{\gamma}_U &= \dfrac{a_U}{V_U} \label{eqn:gM_dot},\\
		\dot{\gamma}_T &= \dfrac{a_T}{V_T}, \label{eqn:gT_dot}
	\end{align}
	where $\theta_U=\gamma_U-\theta$, and  $\theta_T=\gamma_T-\theta$. Here, $R$ is the range to the desired virtual reference point moving on the path, $\theta$ is the LOS angle, $\beta$ is the side slip angle, $V_U$ is the total linear speed of USV, $a_U$ is the lateral acceleration perpendicular to the velocity of the USV, and $\gamma_U$ is the course angle of the USV. Since the virtual reference point, $T$, is moving along the path, the velocity of the virtual point is denoted with $V_T$ along the tangent to the path. The variable $\gamma_T$ denotes the course angle of the moving virtual reference point with respect to the North. Furthermore, it is crucial to state some of the key assumptions and definitions before formally stating the problem addressed in this work.
	
	\begin{assumption}
		The vehicle is assumed to be underactuated with known system dynamics, and no external disturbances are present. 
	\end{assumption}
	
	\begin{assumption}
		Range and LOS information from the desired virtual reference point, the path, $\mathcal{P}$, are available at all times to the control system. In other words, the path information is known or user-defined.
	\end{assumption}
	
	Both assumptions are reasonable from the point of view of application to real-life scenarios. Also, the proposed strategy uses a sliding mode-based control strategy, which is robust to matched disturbances. 
	
	\subsection{Input saturation model with asymmetric bounds }
	Real-world systems often face actuator saturation. In many cases, their capabilities differ in forward and reverse directions, such as the forward and reverse thrust of a USV. While \cite{10.1049/iet-cta.2016.1097} uses a symmetric model with uniform bounds, our previous work \cite{verma2025trajectory} extends this approach to account for asymmetric saturation as follows:
	\begin{equation}\label{eqn:actuator}
		\dot{\pmb{\zeta}} = \pmb{Q}(\pmb{I} - \pmb{G}_M) \pmb{\tau}_c - \pmb{\rho} \pmb{\zeta} + (\pmb{I}-\pmb{Q})(\pmb{I} - \pmb{G}_m) \pmb{\tau}_c,~~\pmb{\tau} = \pmb{\zeta}
	\end{equation}
	where $\pmb{\zeta}(t=0) = \pmb{0}$, $\pmb{\tau}_{c}$, $\pmb{\tau}$, and $\pmb{\zeta} \in \mathbb{R}^3$ are the input, the output, and the state of the actuator saturation model, respectively. The term  $\pmb{\rho} = \diag [\rho_1~~ \rho_2 ~~ \rho_3] $ where $\rho_1, \rho_2, \rho_3 > 0 $ are the design constants. The matrices $\pmb{I},\pmb{G}_M, \pmb{G}_m \in \mathbb{R}^{3 \times 3} $, where $\pmb{I}$ is the identity matrix, $\pmb{G}_M=\diag [ g_{1M} ~~ g_{2M} ~~ g_{3M}]$, where the terms $g_{iM} = \left(\dfrac{\tau_i}{\tau_{iM}}\right)^n\,\forall~i=\{1,2,3\}$, and $\pmb{G}_m=\diag [g_{1m} ~~ g_{2m} ~~ g_{3m}]$, where the terms $g_{im} = \left(\dfrac{\tau_i}{\tau_{im}}\right)^n\,\forall~i=\{1,2,3\}$ and $n\geq 2$ being an even integer. Here, $\tau_{iM}$ and $\tau_{im}$ are the actuator bounds in forward and reverse directions. The matrix $\pmb{Q} = \diag [q_1 ~~q_2~~ q_3]$ helps in enforcing the asymmetric actuator saturation bounds, where 
	\begin{equation}\label{eqn:q_zeta}
		q_i (\zeta_i) = \begin{cases}
			1, & \text{if $\zeta_i$} > 0\\
			0, & \text{$\zeta_i$} \leq 0
		\end{cases}.
	\end{equation}
	The actuator saturation bounds in forward and reverse directions are $\tau_{iM}$ and $\tau_{im}$, respectively. The input saturation model in \Cref{eqn:actuator} ensures that $\pmb{\tau}$ satisfies the actuator bounds while $\pmb{\tau}_c$ can exceed the actuator bounds but is bounded by some finite positive value (see \cite{verma2025trajectory} for the proof). 
	
	Given the assumptions above, we now formally state the problem addressed in this work.
	
	\begin{problem}\label{prob: main} 
		Given the equation of motion of the USV as described by \Cref{eqn: kinematics,eqn: dynamics} and the engagement scenario of the USV with the desired virtual reference point, $T$, as depicted in \Cref{fig:engagement}, the objective of this work is to design a nonlinear integrated guidance and control law for a USV to steer the vehicle to a specified path, $\mathcal{P}$. Mathematically, these objectives translate to 
		\begin{equation*}
			x_T(\Theta) - x(t) \rightarrow 0,~  y_T(\Theta) - y(t) \rightarrow 0,
		\end{equation*}
		as $t \rightarrow \infty$, where variable $\Theta$ is used to parameterize the path.
		Alternatively, in terms of relative motion between a virtual target along the path and the USV position, the objective of path following can be achieved by applying the control inputs to the USV such that
		\begin{equation*}
			\theta_U \rightarrow 0, ~~ R\rightarrow 0, ~~ \dot{R}\rightarrow 0,
		\end{equation*}
		as the time progresses. 
	\end{problem}
	Here, controlling the angle $\theta_U$ to zero will make the USV directly pursue the virtual target $\left(x_T(\Theta),y_T(\Theta)\right)$. Nullification of the range $R$ implies that the USV has reached the path. Additionally, it is crucial to emphasize that merely nullifying $R$ may not be sufficient, as the USV must not only reach the path but also stay on it. Therefore, nullifying $\dot{R}$ is critical. 
	
	To address this problem, we proposed two approaches wherein the first approach we design control inputs without accounting for actuator bounds and later impose the bounds on control inputs using a saturation block. On the other hand, the second approach incorporates the control bounds in the design itself using \Cref{eqn:actuator}. This will allow us to provide a mathematical guarantee of the overall system's stability, even in the presence of actuator constraints.
	
	\section{Main results}\label{sec:main_results}
	In this section, we derive a control strategy to achieve path-following in planar scenarios. The approach is to rendezvous with a virtual reference point moving along the path, resulting in path-following behavior. To do so, we use a pure-pursuit-based strategy. The velocity vector of the USV is first aligned with LOS, and then the relative range and its rate are nullified. Since we formulated the path-following problem as a guidance problem, the commonly used steering command in the interceptor guidance literature is the normal acceleration, $a_U$.  Therefore, we first discuss the equation relating the lateral acceleration of SV to the acceleration of the USV represented in the USV's body frame. This relation was established in our previous work \cite{10.2514/1.G009394, rmv_tpn_acc}, which we now present here as well for completeness.
	
	\begin{lemma}
		The acceleration of USV perpendicular to LOS is related to the yaw rate $r$ and acceleration in body frame $\dot u$ and $\dot v$ with the relation given by 
		\begin{equation}\label{eqn:aM}
			a_U = (-\dot{u}\sin \beta + \dot{v}\cos \beta)+r V_U,
		\end{equation}
		where $\beta$ is the sideslip angle of the USV. 
	\end{lemma}
	\begin{proof}
		To compute USV's acceleration $a_U$, let us express the course angle of USV in terms of sideslip angle ($\beta$) and heading angle ($\psi$).  As illustrated in \Cref{fig:engagement}, the course angle of the USV is given by $\gamma_U = \beta + \psi,$
		which, on differentiating with respect to time, gives
		\begin{equation}\label{eqn: gamma_M_dot}
			\dot{\gamma}_U = \dot{\beta} + \dot{\psi} = \dot{\beta} + r.
		\end{equation}
		Furthermore, the sideslip angle can be obtained as 
		\begin{equation}
			\beta = \tan^{-1}\left(\dfrac{v}{u}\right),
		\end{equation}
		which, on taking the time derivative, results in
		\begin{equation}\label{eqn: beta_dot}
			\dot{\beta} = \dfrac{u\dot{v} - \dot{u}v}{u^2+ v^2} = \dfrac{u\dot{v} - \dot{u}v}{V_U^2}.
		\end{equation}
		Now, on replacing the expression of $\dot{\beta}$ from \Cref{eqn: beta_dot} in \Cref{eqn: gamma_M_dot}, we get
		\begin{align}
			\dot{\gamma}_U = \dfrac{u\dot{v} - \dot{u}v}{V_U^2} + r = \dfrac{1}{V_U}\left(-\dot{u}\dfrac{v}{V_U} + \dot{v}\dfrac{u}{V_U}\right)+r.
		\end{align}
		Noting that $\sin \beta = v/V_U$ and $\cos \beta = u/V_U$ (see \Cref{fig:engagement}), one can rewrite the course rate equation as
		\begin{equation}\label{eqn:gamma_M_dot}
			\dot{\gamma}_U = \dfrac{1}{V_U}\left(-\dot{u}\sin \beta + \dot{v}\cos \beta\right)+r.
		\end{equation}
		Finally, using \Cref{eqn:gM_dot,eqn:gamma_M_dot}, one may reach at \Cref{eqn:aM}. This completes the proof.  
	\end{proof}
	
	Now, we proceed to reorganize the dynamical equations and rewrite them in a compact form. Due to the coupling between the sway and yaw channels, we can control the sway rate even without a sway thruster. To do so, the dynamical equations can be re-derived by substituting $\dot{r}$ from \Cref{eqn:r__dot} into \Cref{eqn:v__dot}, 
	and on eliminating $\dot{r}$, we obtain
	\begin{align}\label{eqn:v__dot2}
		\nonumber \dot{v} &= -\dfrac{m_{23}}{m_{22}}\left( -\dfrac{m_{32}}{m_{33}}\dot{v} + \dfrac{m_{11} - m_{22}}{m_{33}}vu -\dfrac{m_{23}}{m_{33}}ur -\dfrac{d_{32}}{m_{33}}v -\dfrac{d_{33}}{m_{33}}r + \dfrac{\tau_r}{m_{33}} \right) \\ 
		& \quad
		-\dfrac{m_{11}}{m_{22}}ur -\dfrac{d_{22}}{m_{22}}v -\dfrac{d_{23}}{m_{22}}r.
	\end{align}
	On collecting all $\dot{v}$-related terms, we get 
	\begin{align}\label{eqn: v__dot1_}
		\dot{v} &= \dfrac{-m_{23} \left[ (m_{11}-m_{22})vu - m_{23}ur - d_{32}v -d_{33}r + \tau_r\right] }{m_{22}m_{33}-m_{23}m_{32}}
		+ 
		\dfrac{m_{33}\left[ 
			-m_{11}ur -d_{22}v -d_{23}r
			\right]}{m_{22}m_{33}-m_{23}m_{32}}
	\end{align}
	For notational convenience and facilitate the subsequent controller design, we write \Cref{eqn: u__dot,eqn: v__dot1_} as 
	\begin{subequations}
		\begin{align}
			\dot{u} = f_u(\pmb{\nu}) + g_u \tau_u, \label{eqn:udot_compact}\\
			\dot{v} = f_v(\pmb{\nu}) + g_v \tau_r,\label{eqn:vdot_compact}
		\end{align}
	\end{subequations}
	where the terms $f_u(\pmb{\nu})$, $g_u$, $f_v(\pmb{\nu})$, and $g_v$ are given by 
	\begin{subequations}
		\begin{align*}
			f_u(\pmb{\nu}) &= \dfrac{m_{22}}{m_{11}}v r + \dfrac{m_{23}}{m_{11}}r^2 - \dfrac{d_{11}}{m_{11}}u, \quad
			g_u =\dfrac{1}{m_{11}}, \quad  g_v = \dfrac{-m_{23}}{m_{22}m_{33}-m_{23}m_{32}},\\
			f_v(\pmb{\nu}) &= \dfrac{-m_{23}\left[ (m_{11}-m_{22})vu - m_{23}ur  - d_{32}v -d_{33}r\right]}{m_{22}m_{33}-m_{23}m_{32}} 
			+\dfrac{m_{33}\left[ 
				-m_{11}ur -d_{22}v -d_{23}r
				\right]}{m_{22}m_{33}-m_{23}m_{32}}.
		\end{align*}
	\end{subequations}
	Now, we have established the relation between the lateral acceleration ($a_U$) of the USV and the body rates. And from the dynamics of the USV, the body rates are related to the control forces and moments. This facilitates the integrated guidance and control of the USV, which is presented next.
	\subsection{Integrated guidance and control design} \label{sec:IGC}
	In this subsection, we derive the control inputs required to steer the vehicle towards the desired path. Any arbitrary and smooth path can be seen as the locus of virtual target positions. The guidance law should be designed to steer the USV to rendezvous with the virtual target as it moves along the path. This, in turn, will lead to the path-following behavior. Therefore, to achieve our objective, let us consider two sliding surfaces, given by
	\begin{align}
		\label{eqn:Stheta}\mathcal{S}_{\theta_U} &= \theta_U = \gamma_U- \theta, \\
		\label{eqn:SR}\mathcal{S}_R &= \dot{R}+k_R R. 
	\end{align}
	As the value of $\mathcal{S}_{\theta_U}$ approaches zero, the USV will be moving along the LOS. Also, $\mathcal{S}_R$ going to zero will lead to a rendezvous with the virtual reference point. 
	
	Consider a Lyapunov function candidate as $\mathcal{V}_{\theta_U} = 1/2 \mathcal{S}_{\theta_U}^2$. Differentiating $\mathcal{V}_{\theta_U}$ with respect to time gives
	\begin{equation}\label{eqn:V_thetaU_dot}
		\dot{\mathcal{V}}_{\theta_U} = {\mathcal{S}}_{\theta_U}\dot{\mathcal{S}}_{\theta_U} = {\mathcal{S}}_{\theta_U}(\dot{\gamma}_U - \dot{\theta}).
	\end{equation}
	By using $\dot{\theta}$ and $\dot{\gamma}_U$ from \Cref{eqn:r_theta_dot,eqn:gM_dot} into \Cref{eqn:V_thetaU_dot}, $\dot{\mathcal{S}}_{\theta_U}$ can be expressed as 
	\begin{align}\label{eqn:V_thetaM_dot_1}
		\dot{\mathcal{S}}_{\theta_U} &= \left( \dfrac{a_U}{V_U} -\dot{\theta}
		\right) = \dfrac{1}{V_U}\left(-\dot{u}\sin \beta + \dot{v}\cos \beta\right)+r -\dot{\theta} .
	\end{align}
	Further substituting body rates from \Cref{eqn:udot_compact,eqn:vdot_compact}, one can write $\dot{\mathcal{S}}_{\theta_U}$ as 
	\begin{align}
		\dot{\mathcal{S}}_{\theta_U} 
		&= \dfrac{-(f_u(\pmb{\nu}) + g_u \tau_r)\sin \beta + (f_v(\pmb{\nu}) + g_v \tau_r)\cos \beta}{V_U}+r -\dot{\theta} . \label{eqn:S_thetaUdot}
	\end{align}
	Equation \Cref{eqn:S_thetaUdot} can now be rewritten in compact form as 
	\begin{equation}
		\dot{\mathcal{S}}_{\theta_U} = f_{\theta_U} + g_{11}\tau_u + g_{12}\tau_r, \label{eqn:S_thetaU_dot_compact}
	\end{equation}
	where the terms $f_{\theta_U}$, $g_{11}$, and $g_{12}$ are defined as 
	\begin{align*}
		f_{\theta_U} &= \dfrac{-f_u(\pmb{\nu})\sin \beta + f_v(\pmb{\nu})\cos \beta}{V_U} + r-\dot{\theta},\quad
		g_{11} = -\dfrac{g_u \sin \beta}{V_U},\quad
		g_{12} = \dfrac{g_v \cos \beta}{V_U}.
	\end{align*}
	Next, we consider the Lyapunov function candidate as $\mathcal{V}_{R} = 1/2 \mathcal{S}_{R}^2$. On taking time differentiation of $\mathcal{V}_{R}$ yields
	\begin{equation}
		\dot{\mathcal{V}}_{R} = {\mathcal{S}}_{R}\dot{\mathcal{S}}_{R} = {\mathcal{S}}_{R}(\ddot{R} + k_R\dot{R}).\label{eqn:VR_dot}
	\end{equation}
	To find the dynamics of the range rate, $\dot{R}$, we differentiate \Cref{eqn:r_dot} with respect to time, which yields
	\begin{equation}\label{eqn:R_ddot}
		\ddot{R} = \dot{V}_T \cos\theta_T - V_T \sin \theta_T \dot{\theta}_T - 
		\dot{V}_U \cos\theta_U + V_U \sin \theta_U \dot{\theta}_U.
	\end{equation}
	Since \Cref{eqn:R_ddot} depends on $\dot{V}_U$ and $\dot{\theta}_U$, their respective dynamics must be derived. 
	Using the definition of magnitude of the USV's velocity, we know that $V_U^2 = u^2 + v^2$, and to find out $\dot{V}_U$, we differentiate $V_U$ with time as
	\begin{align}
		\nonumber \dot{V}_U V_U = u \dot{u} + v \dot{v}  \implies \dot{V}_U &= \dfrac{u}{V_U} \dot{u} + \dfrac{v}{V_U} \dot{v} , 
	\end{align}
	which leads to 
	\begin{align}
		\dot{V}_U =\dot{u}\cos \beta + \dot{v} \sin \beta.\label{eqn:VU_dot}
	\end{align}
	To derive the dynamics of $\theta_U$, differentiating it with respect to time yields, $\dot{\theta}_U = \dot{\gamma}-\dot{\theta}$. Using  \Cref{eqn:gamma_M_dot,eqn:gM_dot}, we can write
	$V_U\dot{\theta}_U$ as 
	\begin{align}
		V_U\dot{\theta}_U ={a_U} - \dot{\theta}V_U = (-\dot{u}\sin \beta + \dot{v}\cos \beta)+r V_U - \dot{\theta}V_U \label{eqn:VU_thetaU_dot}
	\end{align}
	Substituting \Cref{eqn:VU_dot,eqn:VU_thetaU_dot} into \Cref{eqn:R_ddot}, $\ddot{R}$ can be obtained as
	\begin{align}
		\nonumber \ddot{R} &= \dot{V}_T \cos\theta_T - V_T \sin \theta_T \dot{\theta}_T  - 
		\left(\dot{u}\cos \beta + \dot{v} \sin \beta \right) \cos\theta_U \\
		& \quad + \sin \theta_U \left((-\dot{u}\sin \beta + \dot{v}\cos \beta)+r V_U - \dot{\theta}V_U\right) .\label{eqn:R_ddot1}
	\end{align}
	Finally, using \Cref{eqn:VR_dot,eqn:R_ddot1}, we get the expression for $\dot{\mathcal{S}}_R$ as 
	\begin{align}
		\nonumber \dot{\mathcal{S}}_R &=   k_R\dot{R} + \dot{V}_T \cos\theta_T - V_T \sin \theta_T \dot{\theta}_T - 
		\left(\dot{u}\cos \beta + \dot{v} \sin \beta \right) \cos\theta_U  \\
		\nonumber & \quad +  \sin \theta_U \left(-\dot{u}\sin \beta + \dot{v}\cos \beta+r V_U - \dot{\theta}V_U\right)   \\
		\nonumber &= k_R\dot{R} + \dot{V}_T \cos\theta_T - V_T \sin \theta_T \dot{\theta}_T + rV_U\sin\theta_U\\
		& \quad -\dot{u}\left(\cos \beta \cos\theta_U + \sin \theta_U \sin \beta\right) -\dot{\theta}V_U\sin\theta_U + \dot{v}\left( -\sin \beta \cos\theta_U+\cos \beta \sin\theta_U \right).
		\label{eqn:SR_dot}
	\end{align}
	Now, using \Cref{eqn:udot_compact} and \Cref{eqn:vdot_compact}, we can write \Cref{eqn:SR_dot} as 
	\begin{align}
		\nonumber \dot{\mathcal{S}}_R &=   k_R\dot{R} + \dot{V}_T \cos\theta_T - V_T \sin \theta_T \dot{\theta}_T + rV_U\sin\theta_U  -\dot{\theta}V_U\sin\theta_U\\
		\nonumber & \quad - 
		\left(\cos \beta \cos\theta_U + \sin \theta_U \sin \beta\right) (f_u(\pmb{\nu}) + g_u \tau_u) \\
		& \quad + \left( -\sin \beta \cos\theta_U+\cos \beta \sin\theta_U \right) (f_v(\pmb{\nu}) + g_v \tau_r) \label{eqn:SR_dot1}.
	\end{align}
	For ease of representation, \Cref{eqn:SR_dot1} can also be written in compact form as 
	\begin{equation}
		\dot{\mathcal{S}}_R = g_{21}\tau_u + g_{22} \tau_r + f_R,\label{eqn:SR_dot_compact}
	\end{equation}
	where $g_{21}, g_{22}$, and $f_R$ are defined as 
	\begin{align*}
		g_{21} &= \left(\cos \beta \cos\theta_U + \sin \theta_U \sin \beta\right)g_u,\quad
		g_{22} = \left( -\sin \beta \cos\theta_U+\cos \beta \sin\theta_U \right) g_v,\\
		f_R & = k_R\dot{R} + \dot{V}_T \cos\theta_T - V_T \sin \theta_T \dot{\theta}_T -\left(\cos \beta \cos\theta_U + \sin \theta_U \sin \beta\right)f_u(\pmb{\nu}) \\
		& \quad + \left( -\sin \beta \cos\theta_U+\cos \beta \sin\theta_U \right) f_v(\pmb{\nu}) + rV_U\sin\theta_U  -\dot{\theta}V_U\sin\theta_U
	\end{align*}
	From \Cref{eqn:S_thetaU_dot_compact,eqn:SR_dot_compact}, it can be observed that both equations are coupled in terms of $\tau_u$ and $\tau_r$, which requires them to be obtained simultaneously. Towards this end, consider a combined Lyapunov function candidate as 
	\begin{equation}
		\mathcal{W} = \dfrac{1}{2}\mathcal{S}^\top \mathcal{S}, \quad \mathcal{S} =\begin{bmatrix}
			\mathcal{S}_{\theta_U} &
			\mathcal{S}_R
		\end{bmatrix}^\top
	\end{equation}
	which, on differentiating with respect to time and using \Cref{eqn:S_thetaU_dot_compact,eqn:SR_dot_compact} yields
	\begin{equation} \label{eqn: w_dot_2}
		\dot{\mathcal{W}} = \mathcal{S}^\top\dot{\mathcal{S}}= \begin{bmatrix}
			\mathcal{S}_{\theta_U} &
			\mathcal{S}_R
		\end{bmatrix}\begin{bmatrix}
			\dot{\mathcal{S}}_{\theta_U} \\
			\dot{\mathcal{S}}_{R}
		\end{bmatrix} = \begin{bmatrix}
			\mathcal{S}_{\theta_U} &
			\mathcal{S}_R
		\end{bmatrix} \begin{bmatrix} f_{\theta_U} +g_{11}\tau_u + g_{12}\tau_r\\
			f_R + g_{21}\tau_u + g_{22}\tau_r
		\end{bmatrix}.
	\end{equation}
	On rearranging \Cref{eqn: w_dot_2} and writing in a vectorized form,  
	\begin{equation}
		\dot{\mathcal{W}} = \begin{bmatrix}
			\mathcal{S}_{\theta_U} &
			\mathcal{S}_R
		\end{bmatrix} \left( \begin{bmatrix}
			f_\theta \\
			f_R
		\end{bmatrix} + \begin{bmatrix}
			g_{11} & g_{12}\\
			g_{21} & g_{22}
		\end{bmatrix}\begin{bmatrix}
			\tau_u\\
			\tau_r
		\end{bmatrix}
		\right).
	\end{equation}
	Since, we already know the expressions for $g_{11}$, $g_{12}$, $g_{21}$, and $g_{22}$ and on replacing \begin{align}
		\mathcal{F} &= \begin{bmatrix}
			f_\theta  \\
			f_R
		\end{bmatrix}, \quad \mathcal{G} = \begin{bmatrix}
			g_{11} & g_{12}\\
			g_{21} & g_{22}
		\end{bmatrix} = \begin{bmatrix}
			-\dfrac{g_u \sin \beta}{V_U} & \dfrac{g_v \cos \beta}{V_U} \\
			-g_u \left(c \beta c \theta_U + s \beta s \theta_U\right) & g_v \left(-s\beta c \theta_U + c \beta s \theta_U\right)
		\end{bmatrix},\label{eqn:FG}
	\end{align}
	one can express the term $\dot{\mathcal{W}}$ as
	\begin{equation}
		\dot{\mathcal{W}} = \begin{bmatrix}
			\mathcal{S}_{\theta_U} &
			\mathcal{S}_R
		\end{bmatrix}\begin{bmatrix}
			\mathcal{F} + \mathcal{G}\pmb{\tau}
		\end{bmatrix}.
	\end{equation}
	\begin{remark}\label{rem1}
		The determinant of $\mathcal{G}$ can be obtained as 
		\begin{align*}
			|\mathcal{G}| &= g_ug_v \sin^2 \beta \cos \theta_U + g_ug_v \cos^2 \beta \cos \theta_U -g_ug_v \cos \beta \sin \beta \sin \theta_U 
			+ g_ug_v \cos \beta \sin \beta \sin \theta_U \\
			&= g_u g_v \cos \theta_U,
		\end{align*}
		which is non-zero, except when $\theta_U = \pi/2$ as $g_u$ and $g_v$ are related to the parameters of the USV, which are typically non-zero. The case when $\theta_U=0$ corresponds to the situation when $V_U$ is perpendicular to the LOS. This might happen only momentarily during the engagement, which can be verified by computing $\dot{\theta}_U = \dot{\gamma}_U -\dot{\theta}$, and showing that it is indeed a non-zero quantity at this particular instant. 
		
	\end{remark}
	
	It is clear from the \Cref{rem1} that $\mathcal{G}$ is an invertible matrix, hence by choosing the control input vector as
	\begin{equation}\label{eqn:tau_final}
		\pmb{\tau} = -\mathcal{G}^{-1} \mathcal{F} - \mathcal{G}^{-1}\begin{bmatrix}
			\tilde{M}_\theta\sign(\mathcal{S}_{\theta_U}) + N_\theta \mathcal{S}_{\theta_U}\\
			\tilde{M}_R\sign(\mathcal{S}_R) + N_R \mathcal{S}_R
		\end{bmatrix},
	\end{equation}
	and substituting $\pmb{\tau}$ from \Cref{eqn:tau_final} into $\dot{\mathcal{W}}$, we obtain
	\begin{align}
		\nonumber \dot{\mathcal{W}} &= \begin{bmatrix}
			\mathcal{S}_{\theta_U} &
			\mathcal{S}_R
		\end{bmatrix} \begin{bmatrix}
			- \tilde{M}_\theta\sign(\mathcal{S}_{\theta_U}) - N_\theta \mathcal{S}_{\theta_U}\\
			-\tilde{M}_R\sign(\mathcal{S}_R) -  N_R \mathcal{S}_R
		\end{bmatrix}\\
		\nonumber &=- \tilde{M}_\theta\left |\mathcal{S}_{\theta_U}\right | - \tilde{M}_R\left |S_R\right | -N_\theta \mathcal{S}_{\theta_U}^2 - N_R \mathcal{S}_R^2\\
		&<0~\forall ~\mathcal{S}_{\theta_U}\ne 0, ~\mathcal{S}_R\ne 0 .
	\end{align}
	Here, the controller parameters $\tilde{M}_\theta$, $N_\theta$, $\tilde{M}_R$, and $N_R$, are all positive. The negative definiteness of $\mathcal{\dot{W}}$ guarantees that when control input $\pmb{\tau}$, as given by \Cref{eqn:tau_final}, is applied to the system, the system's states reach the sliding surface $\mathcal{S}=\pmb{0}$. So, from \Cref{eqn:Stheta,eqn:SR}, we have
	\begin{equation*}
		\theta_U = 0, \quad \dot{R} + k_R R = 0,
	\end{equation*}
	which, on solving, leads to
	\begin{equation}\label{eqn: theta_R_time}
		\gamma_U(t) = \theta(t), \quad R(t) = R(0)e^{-k_R t}.
	\end{equation}
	From \Cref{eqn: theta_R_time}, it follows that $\theta_U \rightarrow 0$ and $R\rightarrow 0$ as $t\rightarrow \infty$. The schematic for overall implementation of the strategy is depicted in \Cref{fig:control-schematic}.
	\begin{figure}[h!]
		\centering
		\includegraphics[width=0.8\linewidth]{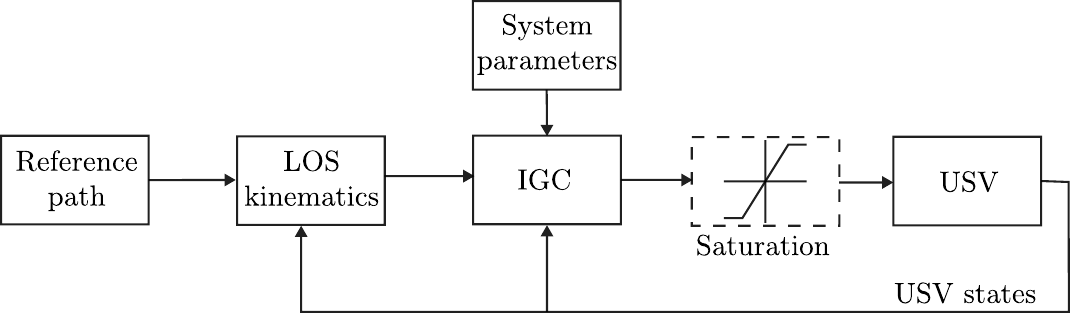}
		\caption{Representation of the overall implementation of IGC design.}
		\label{fig:control-schematic}
	\end{figure}
	So far, we have not considered the actuator bounds in the design. However, as physical actuators have bounded control, it is necessary to bound the control demands. A saturation block is commonly employed, which works in most cases. However, the stability guarantee is no longer applicable. To address this issue related to overall system stability, we incorporate a smooth input saturation model in the design strategy, as presented next.
	\subsection{Design with control input saturation}
	In this subsection, we account for the asymmetric actuator bounds in the design of the control strategy by using the smooth input saturation model in \Cref{eqn:actuator}. In \Cref{sec:IGC}, based on the set objectives and the error variables, we established the dynamics of the switching surfaces. Building upon that, to design the controller while incorporating the input saturation, consider the system dynamics as
	\begin{align}\label{eqn:S_dynamics}
		\begin{bmatrix}
			\dot{\mathcal{S}}_{\theta_U} \\ \dot{\mathcal{S}}_R
		\end{bmatrix}&=\begin{bmatrix} f_{\theta_U} +g_{11}\tau_u + g_{12}\tau_r\\
			f_R + g_{21}\tau_u + g_{22}\tau_r
		\end{bmatrix}.
	\end{align}
	Note that \Cref{eqn:S_dynamics}, is obtained as a result of the combination of USV dynamics as well as the engagement kinematics. Now, the goal is to design control inputs to stabilize these dynamics and meet the path-following objective. To do so, rewriting \Cref{eqn:S_dynamics} and augmenting \Cref{eqn:actuator} to obtain a new transformed systems dynamics as
	\begin{align}
		\dot{\mathcal{S}}&= \mathcal{F}+ \mathcal{G} \mathbf{\tau} \label{eqn:S_vec}\\
		\dot{\pmb{\tau}} &= \left[\pmb{Q}(\pmb{I} - \pmb{G}_M) + (\pmb{I}-\pmb{Q})(\pmb{I} - \pmb{G}_m) \right]\pmb{\tau}_c - \pmb{\rho \tau} \label{eqn:ip_dynamics}
	\end{align}
	where $\mathcal{F}$, $\mathcal{G}$ are as in \Cref{eqn:FG} and $\pmb{\rho} \in \mathbb{R}^{2\times2}$ is a diagonal matrix with positive entries lying between 0 and 1. The closer $\pmb{\rho}$ is to zero, the closer the control inputs approach the saturation bounds. Now, to proceed with the design, define the error variables as 
	\begin{align}
		\pmb{e}_s &= \mathcal{S}, \quad \pmb{z} = \pmb{\tau}- \pmb{\alpha}, \label{eqn:z}
	\end{align}
	where variable $\pmb{z}$ is an intermediate variable and $\pmb{\alpha}$ is a stabilizing function to be designed. Consider a Lyapunov function candidate as
	\begin{align}
		\mathcal{V}_1 = \dfrac{1}{2}\pmb{e}_s^\top \pmb{e}_s = \mathcal{S}^\top \mathcal{S},
	\end{align}
	which, on differentiating with respect to time and using \Cref{eqn:S_vec}, leads to
	\begin{align}
		\dot{\mathcal{V}}_1 = \mathcal{S}^\top \dot{\mathcal{S}} =\mathcal{S}^\top\left[ \mathcal{F}+ \mathcal{G} \mathbf{\tau} \right].
	\end{align}
	On substituting $\pmb{\tau}$ from \Cref{eqn:z} into $\dot{\mathcal{V}}_1$ yields
	\begin{align}
		\dot{\mathcal{V}}_1 =\mathcal{S}^\top\left[ \mathcal{F}+ \mathcal{G} \left( \pmb{z}+\pmb{\alpha}\right) \right]\label{eqn:V1_dot}
	\end{align}
	Choosing the term $\pmb{\alpha}$ as 
	\begin{align}
		\pmb{\alpha} = \mathcal{G}^{-1}\left(-\mathcal{F} - \pmb{K}_1 \mathcal{S} \right). \label{eqn:alpha}
	\end{align}
	Using \Cref{eqn:alpha,eqn:V1_dot}, one may obtain
	\begin{align}
		\dot{\mathcal{V}}_1 =\mathcal{S}^\top \mathcal{G}\pmb{z} - \mathcal{S}^\top\pmb{K}_1 \mathcal{S}.
	\end{align}
	In order to cancel the non-square term, consider another Lyapunov function candidate as
	\begin{align}
		\mathcal{V}_2 =\mathcal{V}_1 +\dfrac{1}{2}\pmb{z}^\top \pmb{z}. \label{eqn:V2}
	\end{align}
	By differentiating \Cref{eqn:V2} with respect to time, we get
	\begin{align}
		\dot{\mathcal{V}}_2 = \dot{\mathcal{V}}_1 + \pmb{z}^\top \dot{\pmb{z}}.\label{eqn:V2_dot}
	\end{align}
	Using \Cref{eqn:z}, one may obtain using time differentiation
	\begin{align}
		\dot{\pmb{z}} = \dot{\pmb{\tau}} - \dot{\pmb{\alpha}} \label{eqn:zdot}
	\end{align}
	Substituting $\dot{\pmb{\tau}}$ from \Cref{eqn:ip_dynamics}, we get
	\begin{align}
		\dot{\pmb{z}} =\left[\pmb{Q}(\pmb{I} - \pmb{G}_M) + (\pmb{I}-\pmb{Q})(\pmb{I} - \pmb{G}_m) \right]\pmb{\tau}_c - \pmb{\rho \tau} - \dot{\pmb{\alpha}} \label{eqn:z_dot}
	\end{align}
	Choosing the command input $\pmb{\tau}_c$ as 
	\begin{align}
		\nonumber \pmb{\tau}_c &= \left[\pmb{Q}(\pmb{I} - \pmb{G}_M) + (\pmb{I}-\pmb{Q})(\pmb{I} - \pmb{G}_m) \right]^{-1}\left( \pmb{\rho \tau} + \dot{\pmb{\alpha}}\right) - \mathcal{G}^\top \mathcal{S} - \pmb{K}_2\pmb{z}\label{eqn:tauc}
	\end{align}
	where $\pmb{K}_2 = \diag [k_{21}, k_{22}]$ satisfying $k_{2i} >0 ~\forall ~i \in \{1,2\}$
	Using \Cref{eqn:tauc,eqn:z_dot,eqn:V2_dot}, we get
	\begin{align}
		\dot{\mathcal{V}}_2 =  - \mathcal{S}^\top\pmb{K}_1 \mathcal{S} -\pmb{z}^\top \pmb{K_2z}.
	\end{align}
	As the constants $k_{1i}$ and $k_{2i}$ for $i \in \{1,2\}$ are all positive, the derivative of $\mathcal{V}_2$ becomes negative definite. Therefore, the stability of the derived control is guaranteed with errors converging to zero. The error variable $\mathcal{S}$ going to zero leads to path-following. The gain parameters are to be chosen such that the USV first nullifies its error with LOS and then reaches the path. 
	\begin{figure}[h!]
		\centering
		\includegraphics[width=0.8\linewidth]{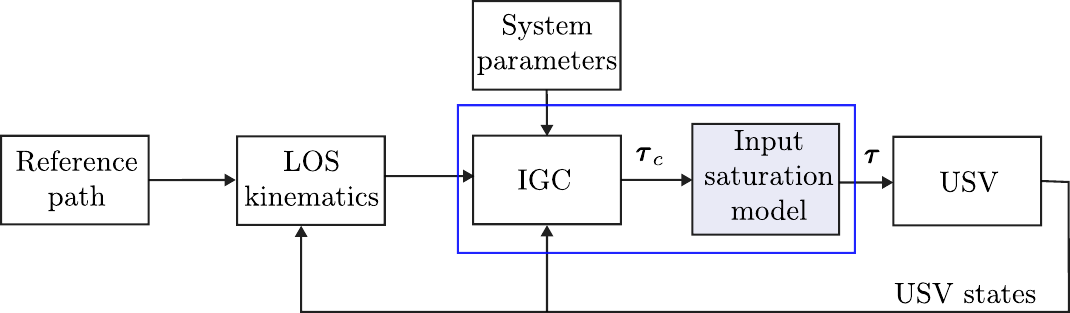}
		\caption{Schematic of implementation of IGC with the consideration of input saturation.}
		\label{fig:control-schematic-inp-sat}
	\end{figure}
	\Cref{fig:control-schematic-inp-sat} shows a schematic of the implementation of the control strategy while accounting for the actuator bounds in the design.
	\section{Simulations}\label{sec:simulation}
	In this section, we validate and demonstrate the performance of the proposed integrated path-following algorithm for the USVs via MATLAB\textsuperscript{\textregistered} simulations. In all the simulations, we have used the parameters of the CyberShip II model adopted from \cite{skjetne2004modeling}. The model is $1:70$ scaled version of a supply vessel with the model dimensions as $(1.255\,\rm{m} \times 0.29\,\rm{m})$. The hydrodynamic derivatives and other model parameters are tabulated in \Cref{tab:params}. For performance validation, we consider elliptical and 8-shaped paths, which are sufficiently smooth and complex.
	\begin{table}[htbp]
		\centering
		\caption{Model parameters of the CyberShip II \cite{skjetne2004modeling}.}
		\begin{tabular}{|c|c|c|c|c|c|}
			\hline
			\textbf{Params} & \textbf{Value} & \textbf{Params} & \textbf{Value} & \textbf{Params} & \textbf{Value} \\
			\hline
			$m$ & 23.800 & $I_z$ & 1.760 & $x_g$ & 0.046 \\
			$X_{\dot{u}}$ & -2.0 & $Y_{\dot{v}}$ & -10.0 & $Y_{\dot{r}}$ & -0.0 \\
			$N_{\dot{v}}$ & -0.0 & $N_{\dot{r}}$ & -0.0 & $X_u$ & -0.72253 \\
			$X_{|u|u}$ & -1.32742 & $X_{uuu}$ & -5.86643 & $Y_{v}$ & -2.0 \\
			$Y_{|v|v}$ & -36.47287 & $N_{v}$ & 0.03130 & $N_{|v|v}$ & 3.95645 \\
			$Y_{|r|v}$ & -0.805 & $Y_{r}$ & -7.250 & $Y_{|v|r}$ & -0.845 \\
			$Y_{|r|r}$ & -3.450 & $N_{|r|v}$ & 0.130 & $N_{r}$ & -1.900 \\
			$N_{|v|r}$ & 0.080 & $N_{|r|r}$ & -0.750 & & \\
			\hline
		\end{tabular}
		\label{tab:params}
	\end{table}
	
	\subsection{Without input saturation model}
	The controller parameters are taken as $k_R = 5$, $\tilde{M}_\theta = 0.3$, $\tilde{M}_R = 0.08$, $N_R = 0.08$, and $N_\theta = 0.3$. Note that the controller gains are chosen by systematically varying the parameters such that $\theta_U$ goes to zero quickly and then $R\rightarrow0$ is achieved. Observe that gain values with subscript $(\cdot)_\theta$, which are responsible for steering the heading, are slightly higher than ones subscripted with $(\cdot)_R$, responsible for reducing the distance between the path and the USV.
	
	\subsubsection{Elliptical reference path}
	First, we present the path-following for the case of an elliptical path described by the equation
	\begin{align}
		x_d(\Theta) &= 4 \sin \Theta, \quad y_d(\Theta) = 2.5(1- \cos \Theta),\label{eqn:ellipse}
	\end{align}
	where the path evolves with time according to $\Theta= 0.05\,t$. The USV is positioned into three different initial configurations denoted with $P_1$, $P_2$, and $P_3$ in \Cref{fig:ellipse}. In the case of $P_1$, the USV's initial location is $2$ m South and $5$ m West with its reference axis (course angle, $\psi$) aligned with North at $30^\circ$. Next, for case $P_2$, the USV is positioned at $3$ m South and $3$ m East with initial $\psi = -30^\circ$. In the last case, $P_3$, the USV is placed $6$ m North and $4$ m West and $\psi = 140^\circ$. In each of these cases, the USV is initially moving at a surge velocity of $0.5$ m/s. Note that in all subsequent simulation figures illustrating the path, the initial position of the USV is denoted with magenta filled square while the starting point of the reference path is marked with black filled square.
	\begin{figure*}[ht!]
		\centering
		\begin{subfigure}{0.5\textwidth}
			\includegraphics[width=\linewidth]{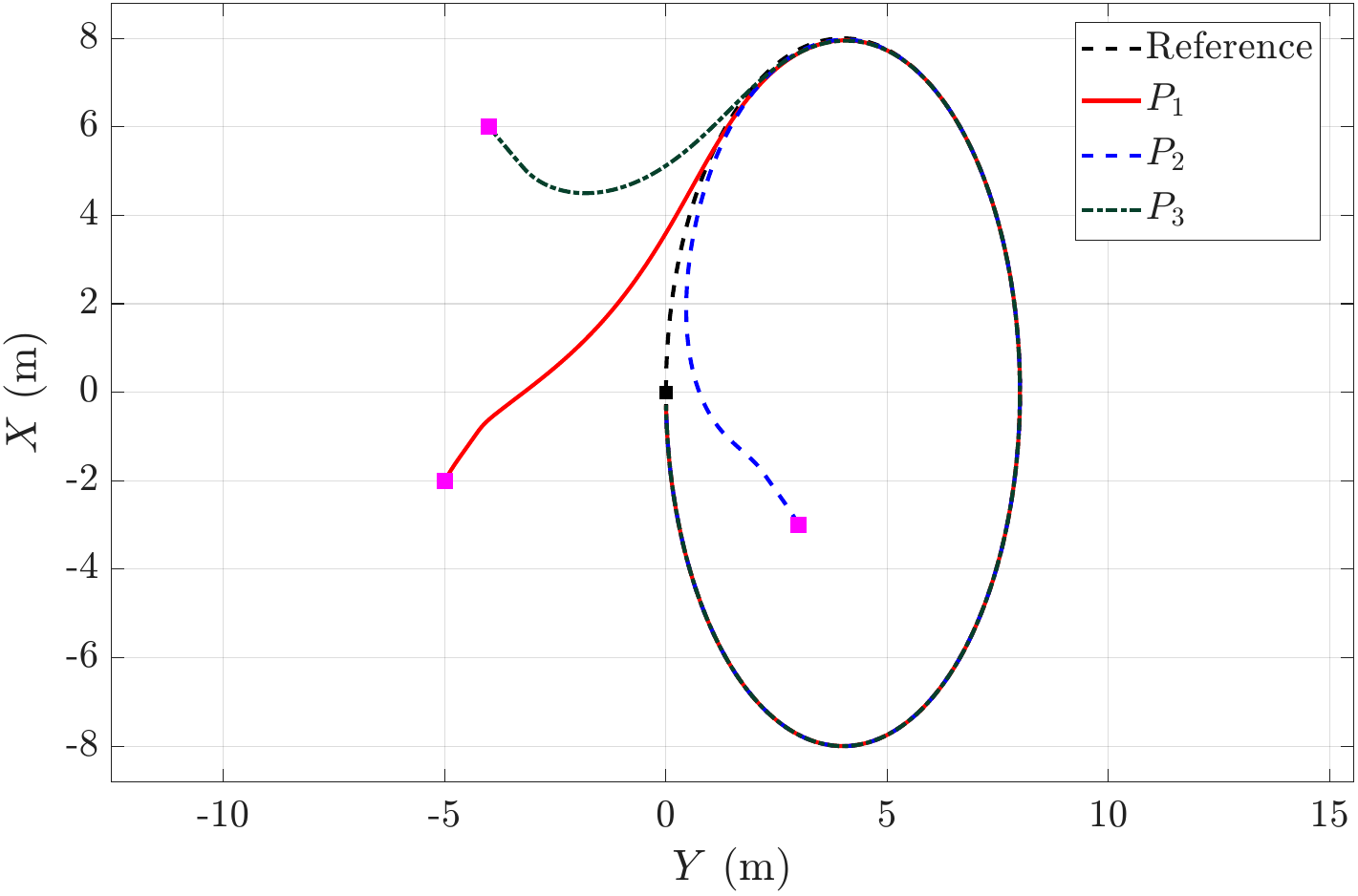}
			\caption{Path of the USV in inertial frame.}
			\label{fig:ellipse_path}
		\end{subfigure}%
		\begin{subfigure}{0.5\textwidth}
			\includegraphics[width=\linewidth]{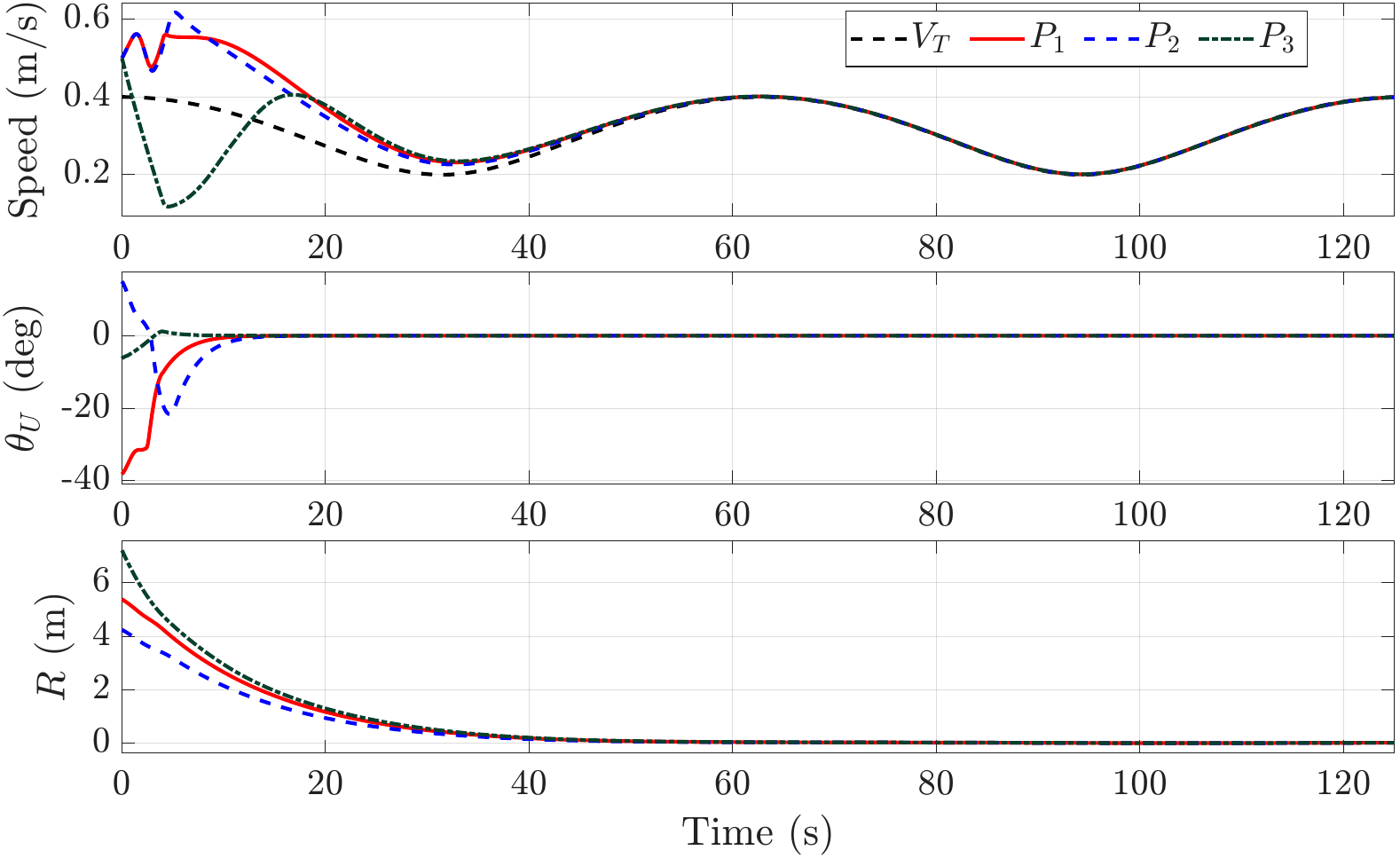}
			\caption{Speed, $\theta_U$ and $R$.}
			\label{fig:ellipse_speed_lead_R}
		\end{subfigure}
		\begin{subfigure}{0.5\textwidth}
			\includegraphics[width=\linewidth]{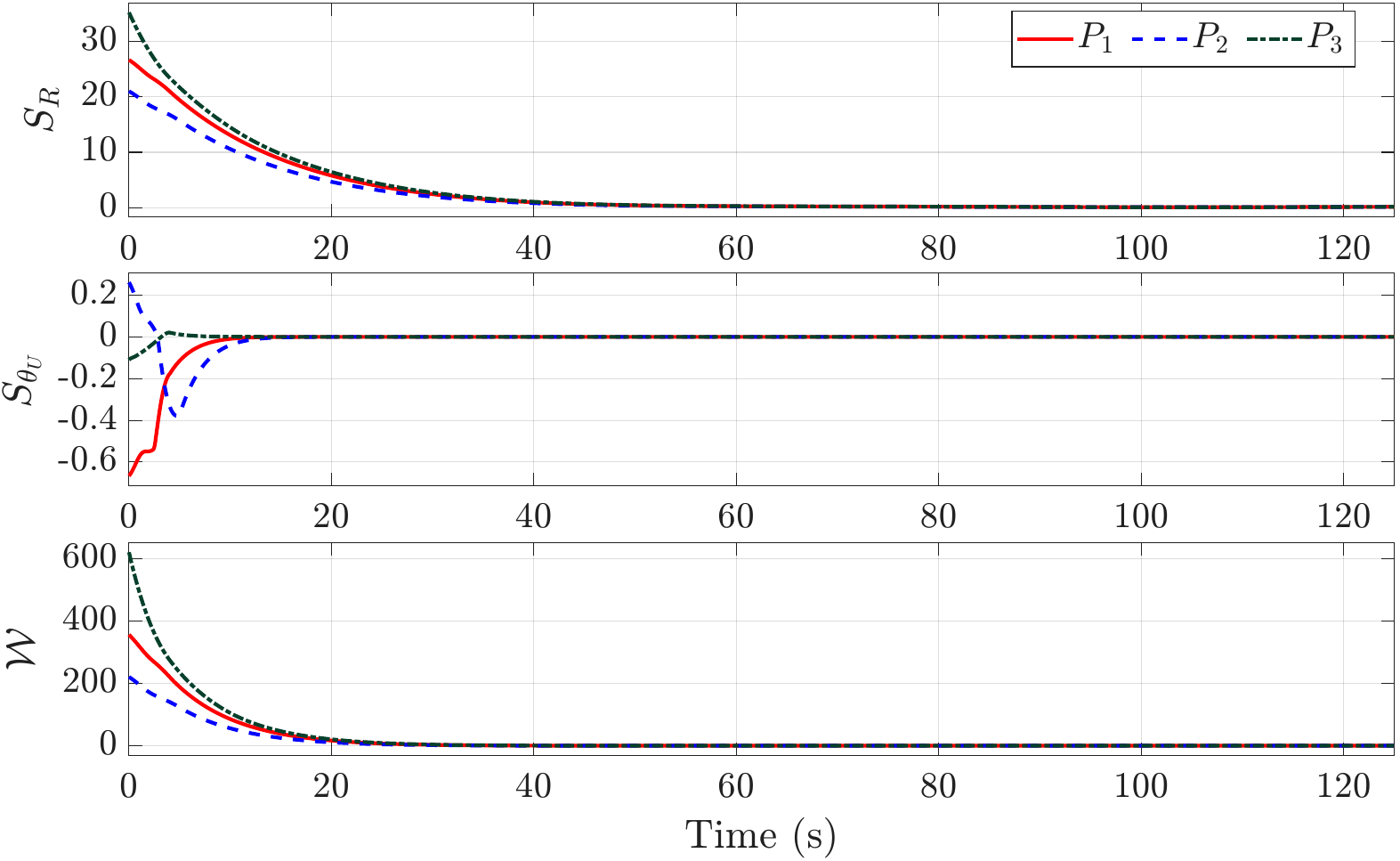}
			\caption{Evolutions of sliding surfaces and Lyapunov function.}
			\label{fig:ellipse_surface}
		\end{subfigure}%
		\begin{subfigure}{0.5\textwidth}
			\includegraphics[width=\linewidth]{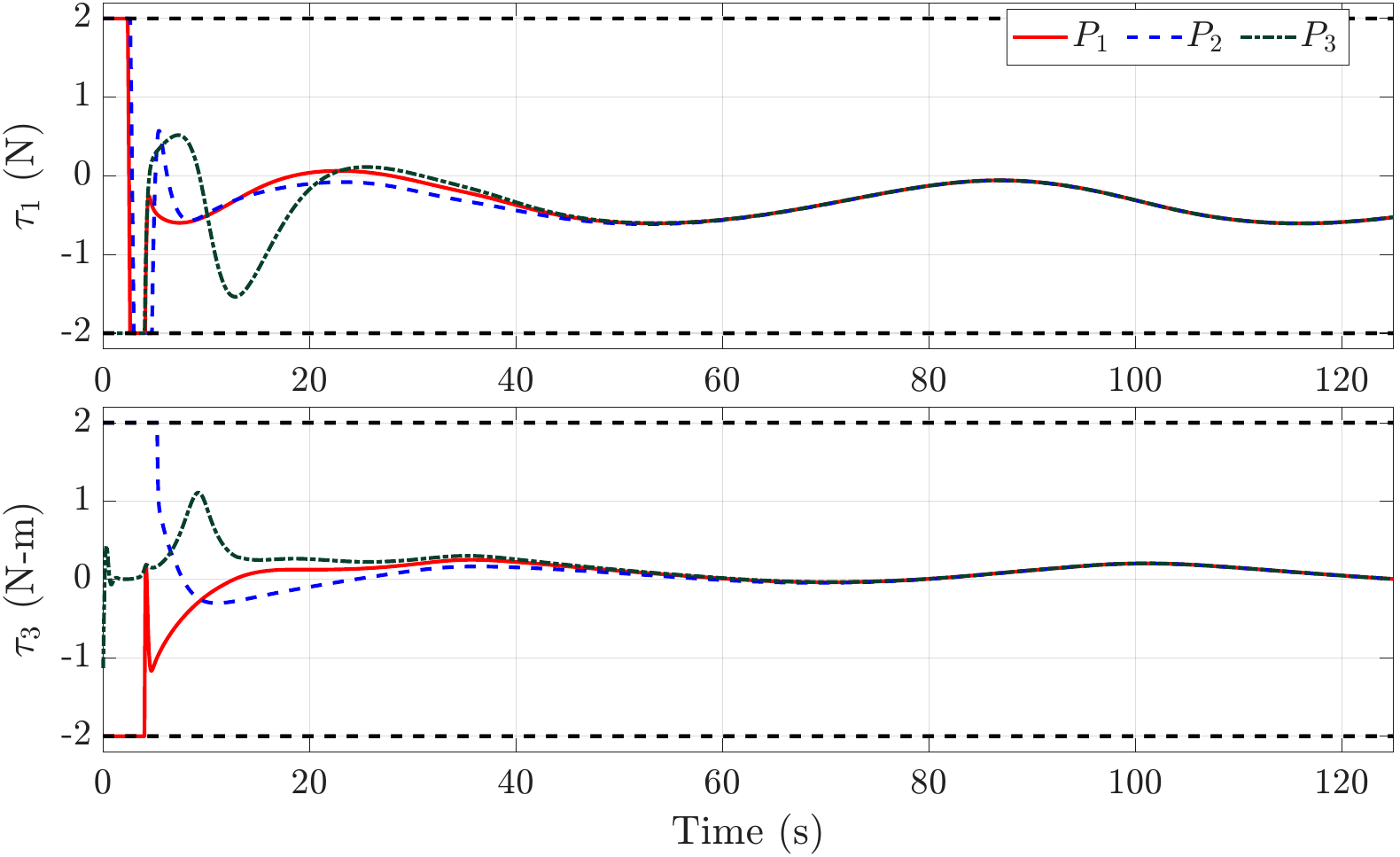}
			\caption{Control inputs.}
			\label{fig:ellipse_tau}
		\end{subfigure}
		\caption{Path-following performance for the case of an elliptical reference path.}
		\label{fig:ellipse}
	\end{figure*} 
	As shown in \Cref{fig:ellipse_path}, USV tracks the reference path in all the cases. It can be observed from \Cref{fig:ellipse_speed_lead_R}, that the USV first aligns its velocity vector with LOS, making $\theta_U$ zero. After that, the range $R$ converges to zero. Also, the speed of the USV tracks the reference speed ($V_T$).
	This clearly demonstrates that USV can follow the reference path using the proposed controller. Also, the sliding surfaces, along with the value of the Lyapunov function, as depicted in  \Cref{fig:ellipse_surface}, decrease to zero. Though it should decrease monotonically, the initial deviations are due to the saturation in the control commands. \Cref{fig:ellipse_tau} shows the control input being applied to the USV. The initial inputs saturate as the USV quickly steers its heading and gets closer to the path. Later, the control demand reduces once the alignment is done.
	
	\subsubsection{8-shaped reference path}
	Next, we illustrate the performance of the proposed algorithm described by \Cref{eqn:tau_final}, in following another challenging 8-shaped path, governed by the equation 
	\begin{align}
		x_d = 8 \cos \Theta_1 -4, \quad y_d = 4 \sin \Theta_2,\label{eqn:eight}
	\end{align}
	where the evolution of path parameters is according to $\Theta_1 = 0.05t$, and $\Theta_2 = 0.1$. 
	Now, we validate the proposed path-following strategy across three different initial configurations of the USV. The initial configuration of the USV is described by its position and heading as indicated by the tuple, $(x,y,\psi)$. In the case of $P_1$, the US starts from $(5~\text{m},0~\text{m},120^\circ)$. For $P_2$, the USV is initially at $(2~\text{m},-2~\text{m},70^\circ)$. In the last case, $P_3$, the initial configuration of the USV is $(5~\text{m},-3~\text{m},100^\circ)$. In all these cases, the USV's initial velocity is taken to be $\pmb{\nu}(0) = (0.5~\text{m/s},0~\text{m/s},0~\text{rad/s})$.
	\begin{figure*}[ht]
		\centering
		\begin{subfigure}{0.5\textwidth}
			\includegraphics[width=\linewidth]{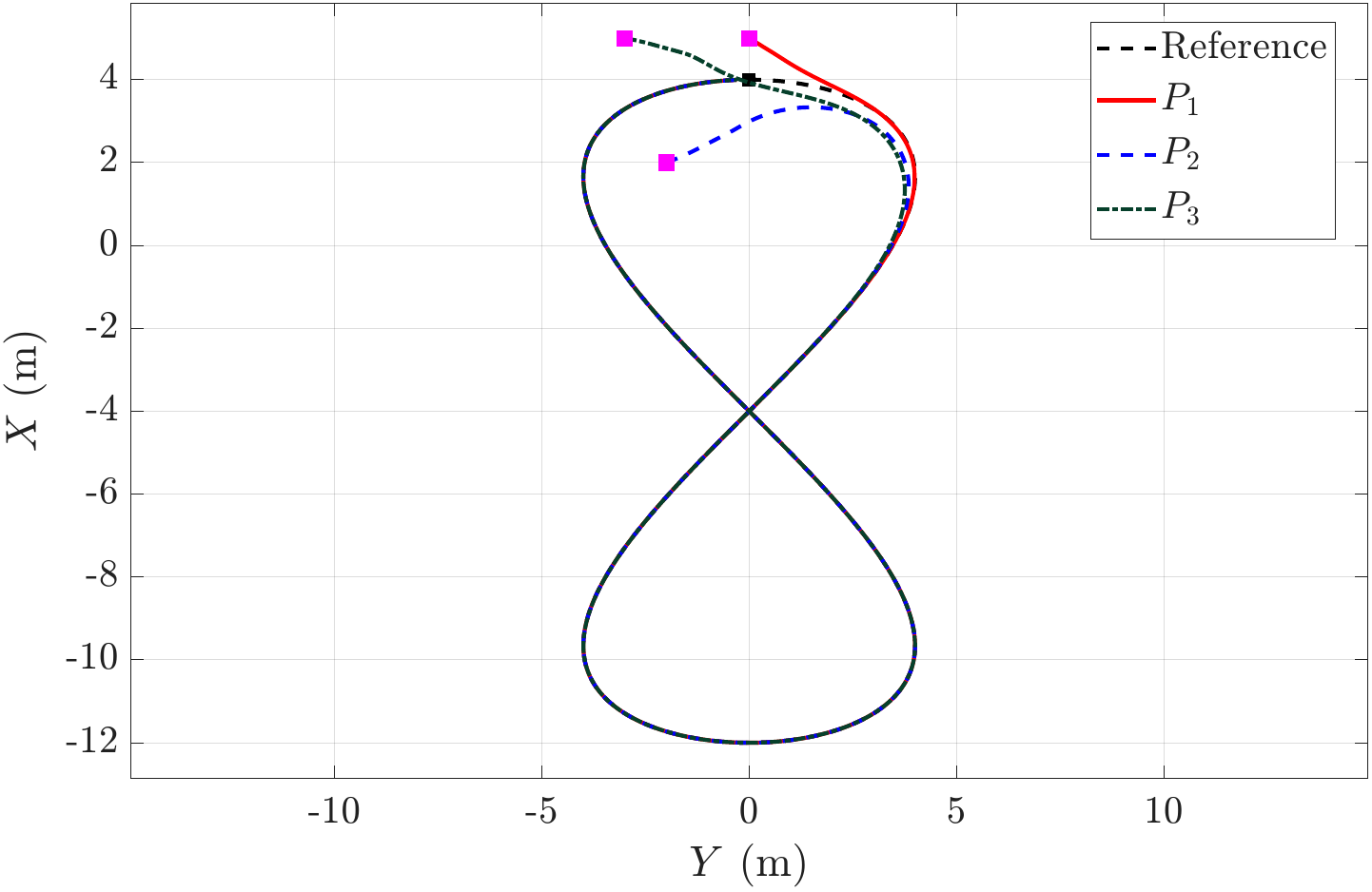}
			\caption{Path of the USV in inertial frame.}
			\label{fig:eight_path}
		\end{subfigure}%
		\begin{subfigure}{0.5\textwidth}
			\includegraphics[width=\linewidth]{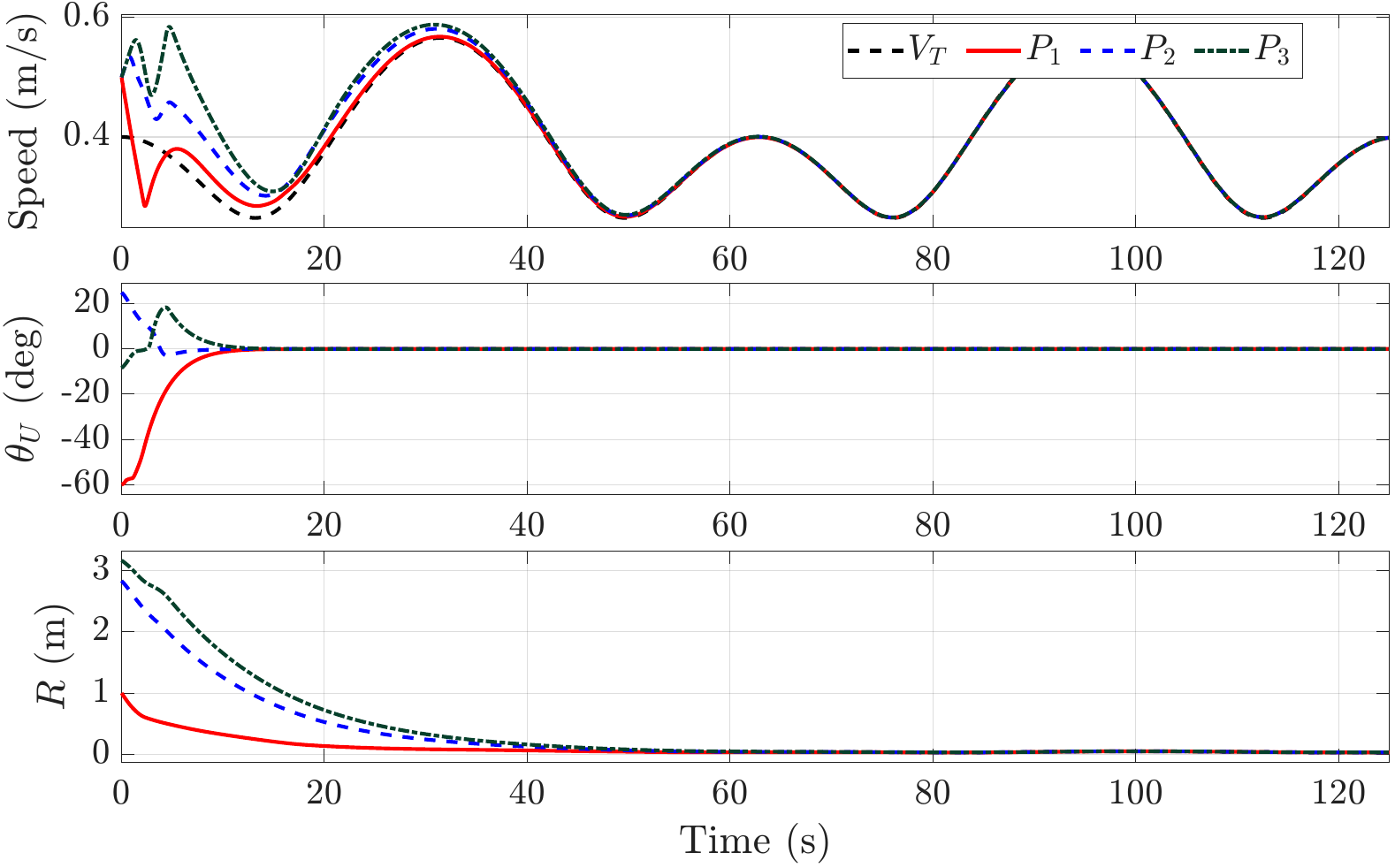}
			\caption{Speed, $\theta_U$ and $R$.}
			\label{fig:eight_speed_lead_R}
		\end{subfigure}
		\begin{subfigure}{0.5\textwidth}
			\includegraphics[width=\linewidth]{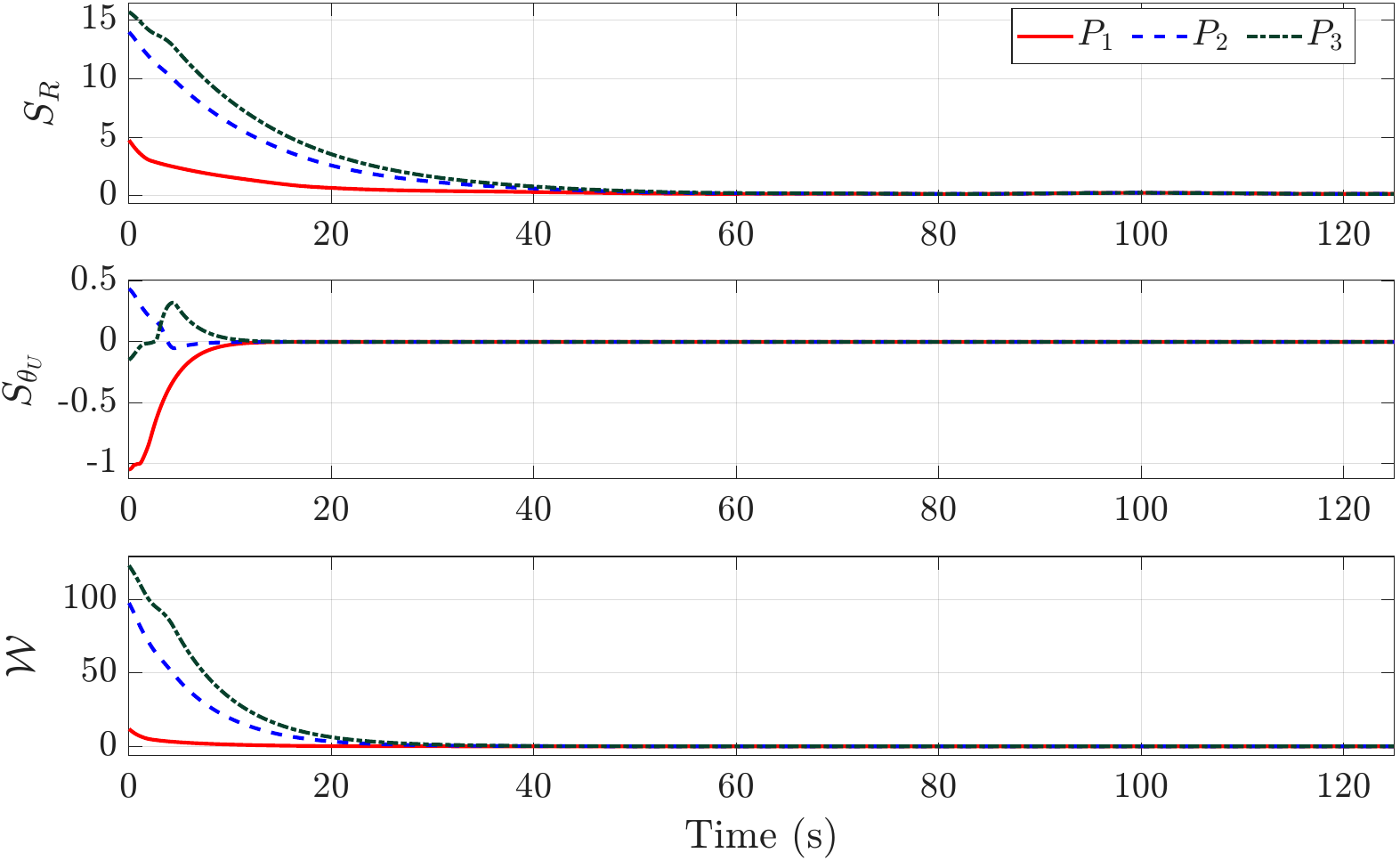}
			\caption{Evolutions of sliding surfaces and Lyapunov function.}
			\label{fig:eight_surface}
		\end{subfigure}%
		\begin{subfigure}{0.5\textwidth}
			\includegraphics[width=\linewidth]{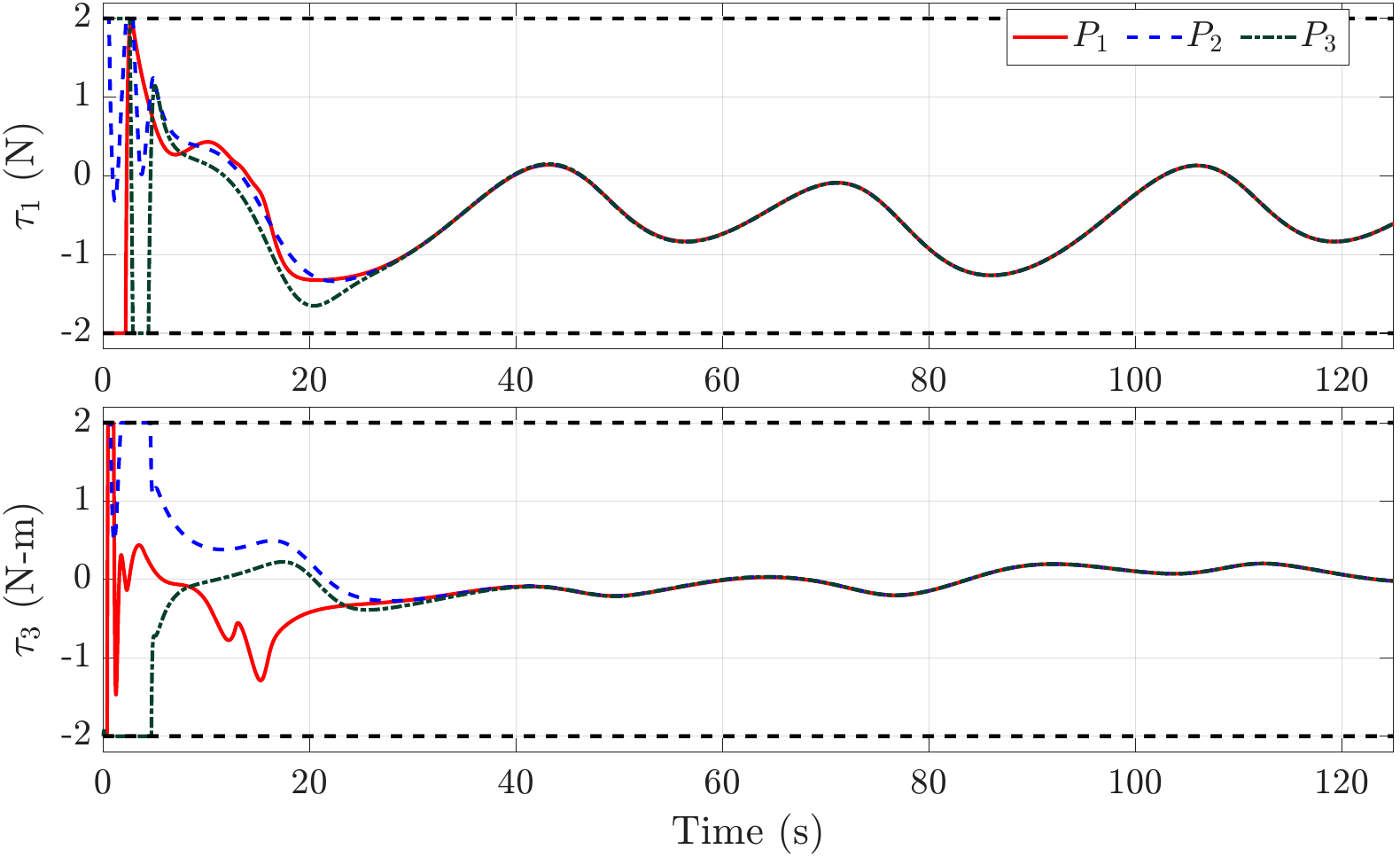}
			\caption{Control inputs.}
			\label{fig:eight_tau}
		\end{subfigure}
		\caption{Path-following performance for the case of an 8-shaped reference path.}
		\label{fig:eight}
	\end{figure*}
	\Cref{fig:eight} shows the performance of the proposed integrated guidance algorithm in following an 8-shaped path. In \Cref{fig:eight_path}, the USV follows the reference path accurately, starting from three different initial conditions denoted with $P_1$, $P_2$, and $P_3$. As evidenced by \Cref{fig:eight_speed_lead_R}, the speed of the USV closely tracks the reference speed. Additionally, the error in aligning the USV's velocity vector with the LOS approaches zero. Gradually, the distance between the virtual reference point and the path also approaches zero. This is also evidenced by the plots of Evolutions of sliding surfaces and Lyapunov function values shown in \Cref{fig:eight_surface}. Thus, it is demonstrated that the USV can steer itself from different initial positions to the desired reference path and maintain it on the path. The control input profile, as depicted in \Cref{fig:eight_tau}, is similar to the case of an elliptical path. The initial control input saturation is due to large initial errors that subsequently decrease, leading to a reduction in the control input demand. As the USV follows the path with different speeds assigned to each point, it steers the speed by adjusting the surge thrust, $\tau_1$. Hence, it is observed that the speed and surge thrust profiles are similar.
	
	\subsection{With input saturation model}
	In this subsection, we validate the control strategy proposed in \Cref{eqn:ip_dynamics,eqn:tauc}. This strategy incorporates the bounds on the control input into the design of the control inputs.  The controller parameters used for simulations are given as: $\pmb{\rho} = \diag [0.2~~ 0.2]$, $\pmb{K}_1 = \diag[0.2~~0.1]$, $\pmb{K}_2 = \diag [5~~1]$. We test the efficacy of the proposed algorithm through numerical simulations in following elliptical and 8-shaped reference paths.
	
	\subsubsection{Elliptical reference path}
	The elliptical reference path is governed by \Cref{eqn:ellipse}. An illustration of the path is shown in \Cref{fig:inp_sat_ellipse_path}.
	\begin{figure*}[ht!]
		\centering
		\begin{subfigure}{0.5\textwidth}
			\includegraphics[width=\linewidth]{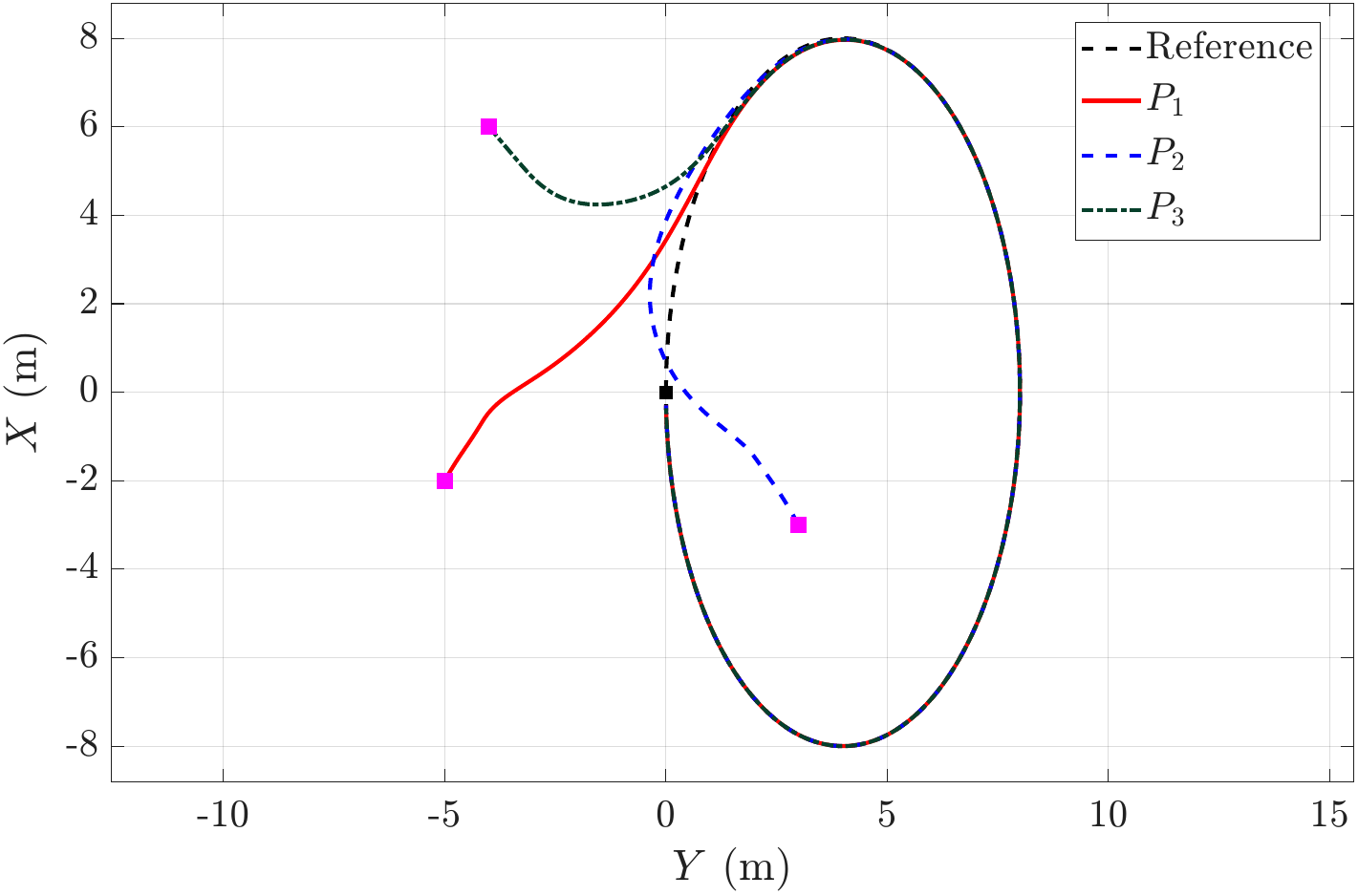}
			\caption{Path of the USV in inertial frame.}
			\label{fig:inp_sat_ellipse_path}
		\end{subfigure}%
		\begin{subfigure}{0.5\textwidth}
			\includegraphics[width=\linewidth]{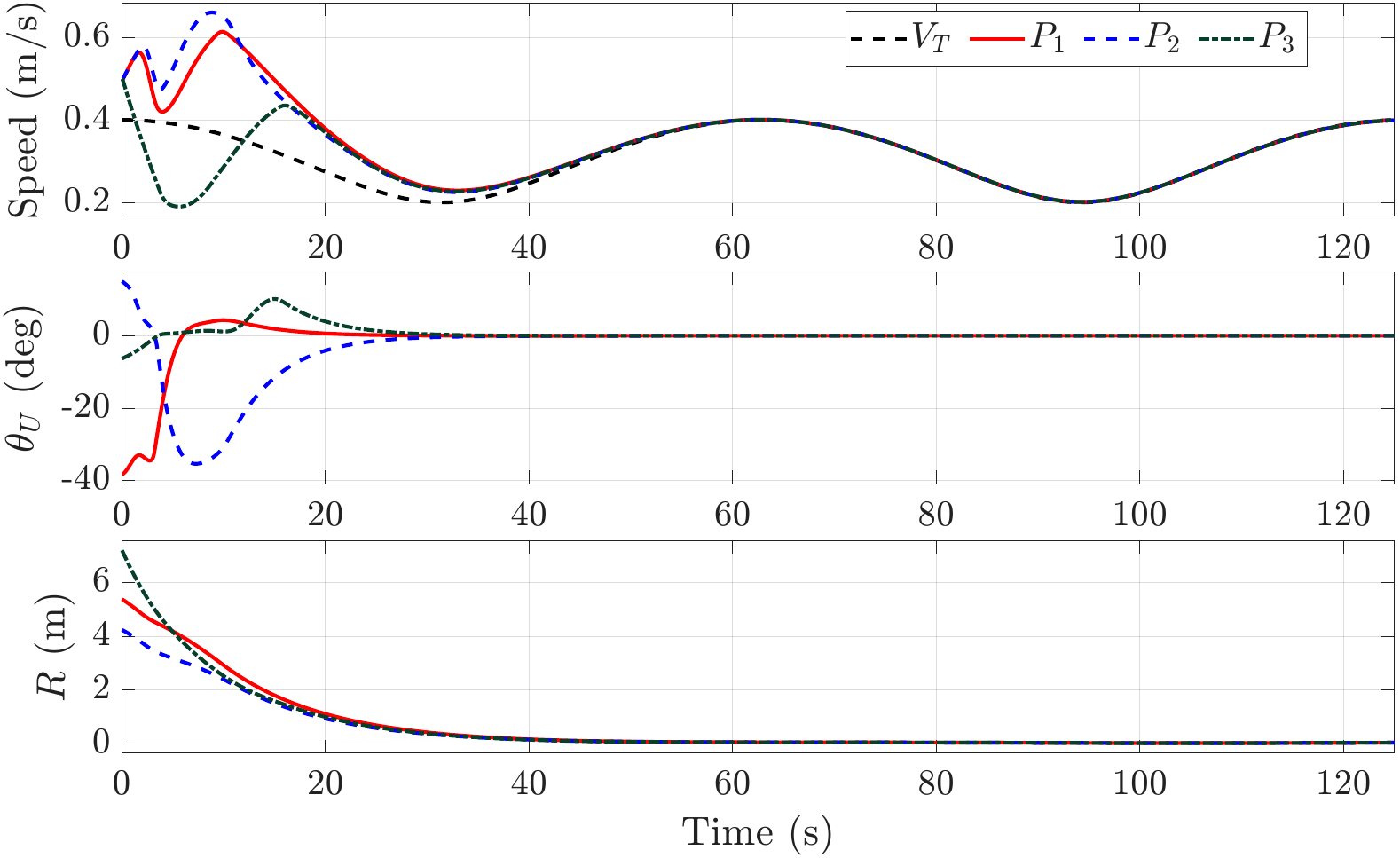}
			\caption{Speed, $\theta_U$ and $R$.}
			\label{fig:inp_sat_ellipse_speed_lead_R}
		\end{subfigure}
		\begin{subfigure}{0.5\textwidth}
			\includegraphics[width=\linewidth]{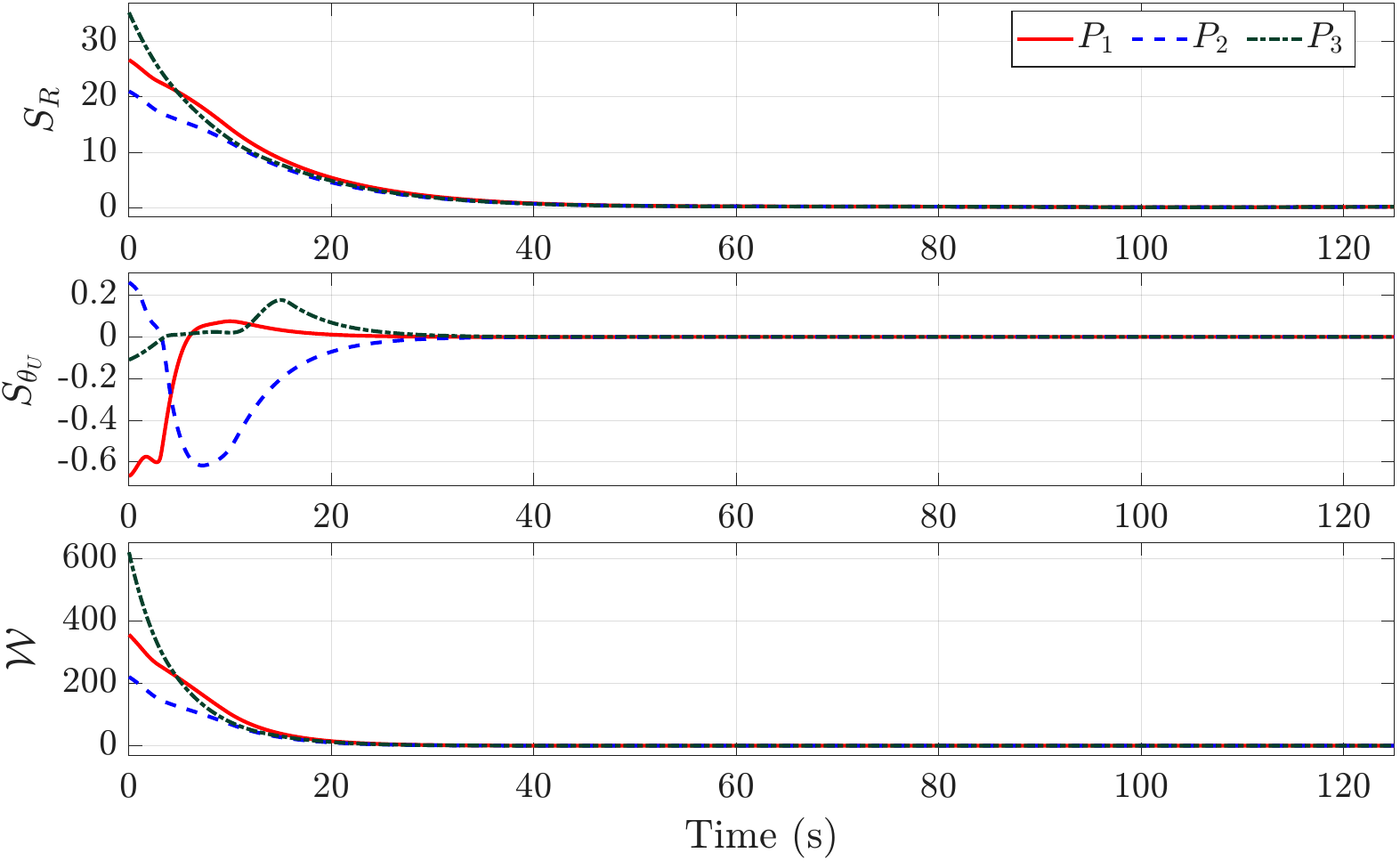}
			\caption{Evolutions of sliding surfaces and Lyapunov function.}
			\label{fig:inp_sat_ellipse_surface}
		\end{subfigure}%
		\begin{subfigure}{0.5\textwidth}
			\includegraphics[width=\linewidth]{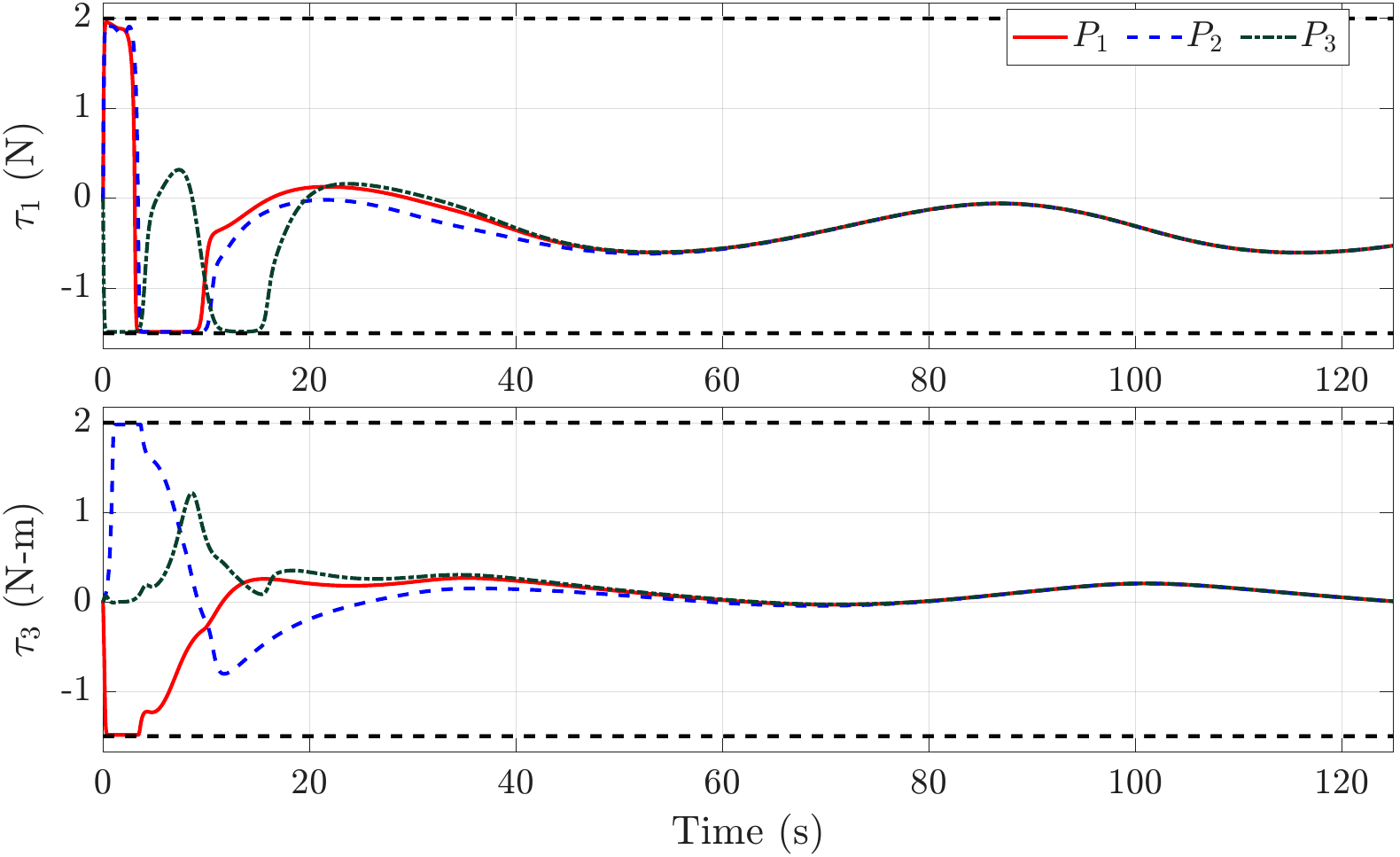}
			\caption{Control inputs.}
			\label{fig:inp_sat_ellipse_tau}
		\end{subfigure}
		\caption{Path-following performance for the case of an elliptical reference path under input saturation.}
		\label{fig:inp_sat_ellipse}
	\end{figure*}
	As evidenced from \Cref{fig:inp_sat_ellipse}, the USV could follow the reference path accurately. \Cref{fig:inp_sat_ellipse_path} shows that the USV follows the path from different chosen initial conditions, namely $P_1$, $P_2$, and $P_3$. The initial configurations in the cases $P_1$, $P_2$, and $P_3$ are denoted by a triplet as $(x,\,y,\,\psi)$: $(-2~\text{m},-5~\text{m},30^\circ)$, $(-3~\text{m},3~\text{m},-30^\circ)$, and $(6~\text{m},-4~\text{m},140^\circ)$.  For all these cases, the USV is able to align itself with LOS, reach the path, and catch up with the assigned speed (see \cref{fig:inp_sat_ellipse_speed_lead_R}). Looking at the speed profile of the USV in \Cref{fig:inp_sat_ellipse_speed_lead_R}, especially for the case $P_3$, when the USV nears the virtual reference point on the path, it decreases its speed (see around $5$ sec). 
	But, now that the reference point is moving at a higher speed, the USV increases its speed and slowly matches the speed assigned to the virtual reference point on the path. This affirms the expected behavior, and USV could follow the specified reference path.
	The control input profiles are depicted in \Cref{fig:inp_sat_ellipse_tau}, which clearly shows that they never exceed the bounds. Note that the asymmetric input bounds are indicated with black dotted horizontal lines. Furthermore, the control input profiles are now smoother and free from corners as opposed to the case when actuator bounds are enforced using a saturation block.
	Once the USV converges on the path, the control input demands in the cases appear to overlap closely.
	
	\subsubsection{8-shaped reference path}
	We now consider an 8-shaped reference path governed by \Cref{eqn:eight} for the demonstration of the efficacy of the proposed technique, which incorporates the actuator bounds in the IGC design.
	\begin{figure*}[ht!]
		\centering
		\begin{subfigure}{0.5\textwidth}
			\includegraphics[width=\linewidth]{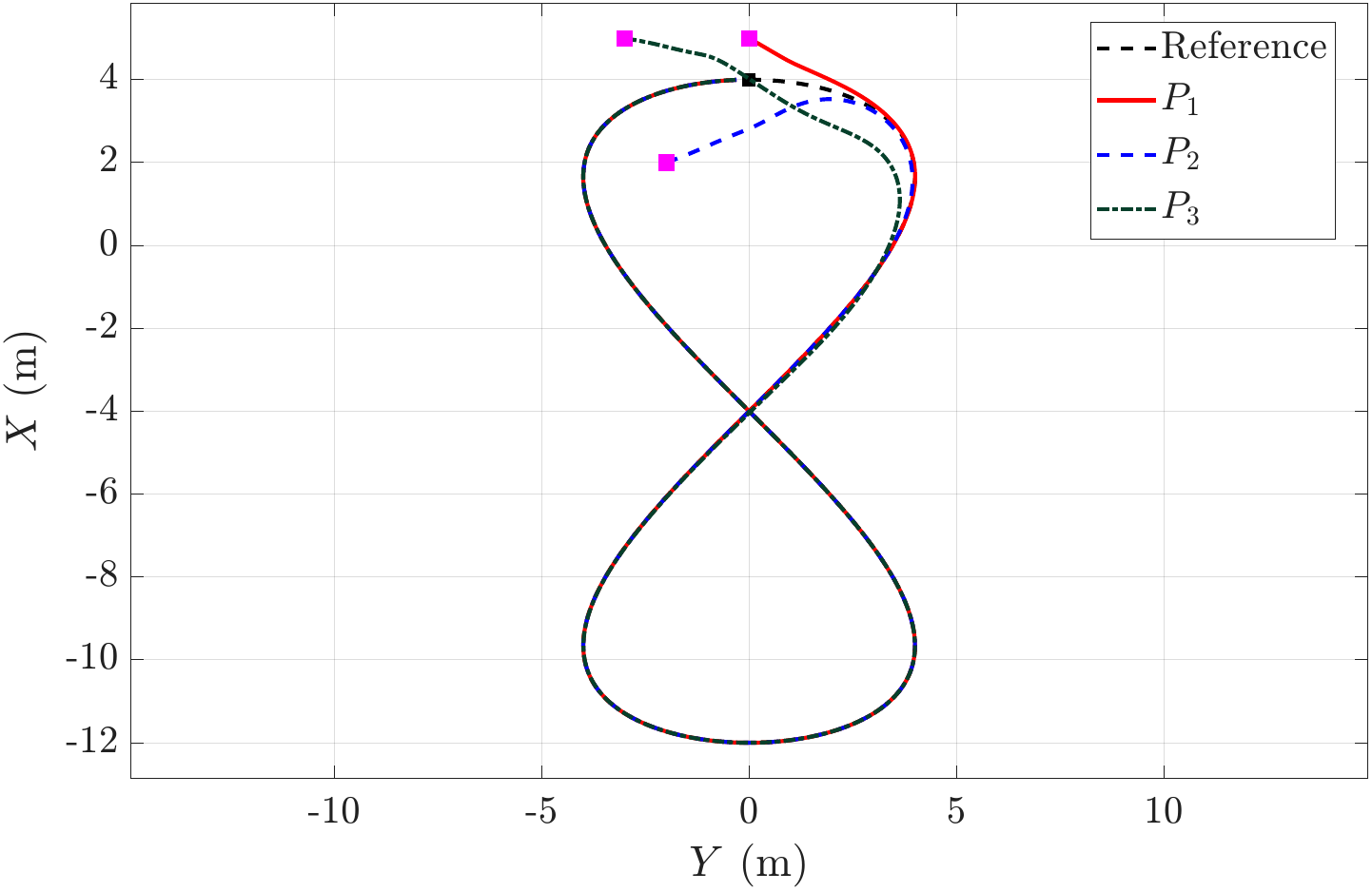}
			\caption{Path of the USV.}
			\label{fig:inp_sat_eight_path}
		\end{subfigure}%
		\begin{subfigure}{0.5\textwidth}
			\includegraphics[width=\linewidth]{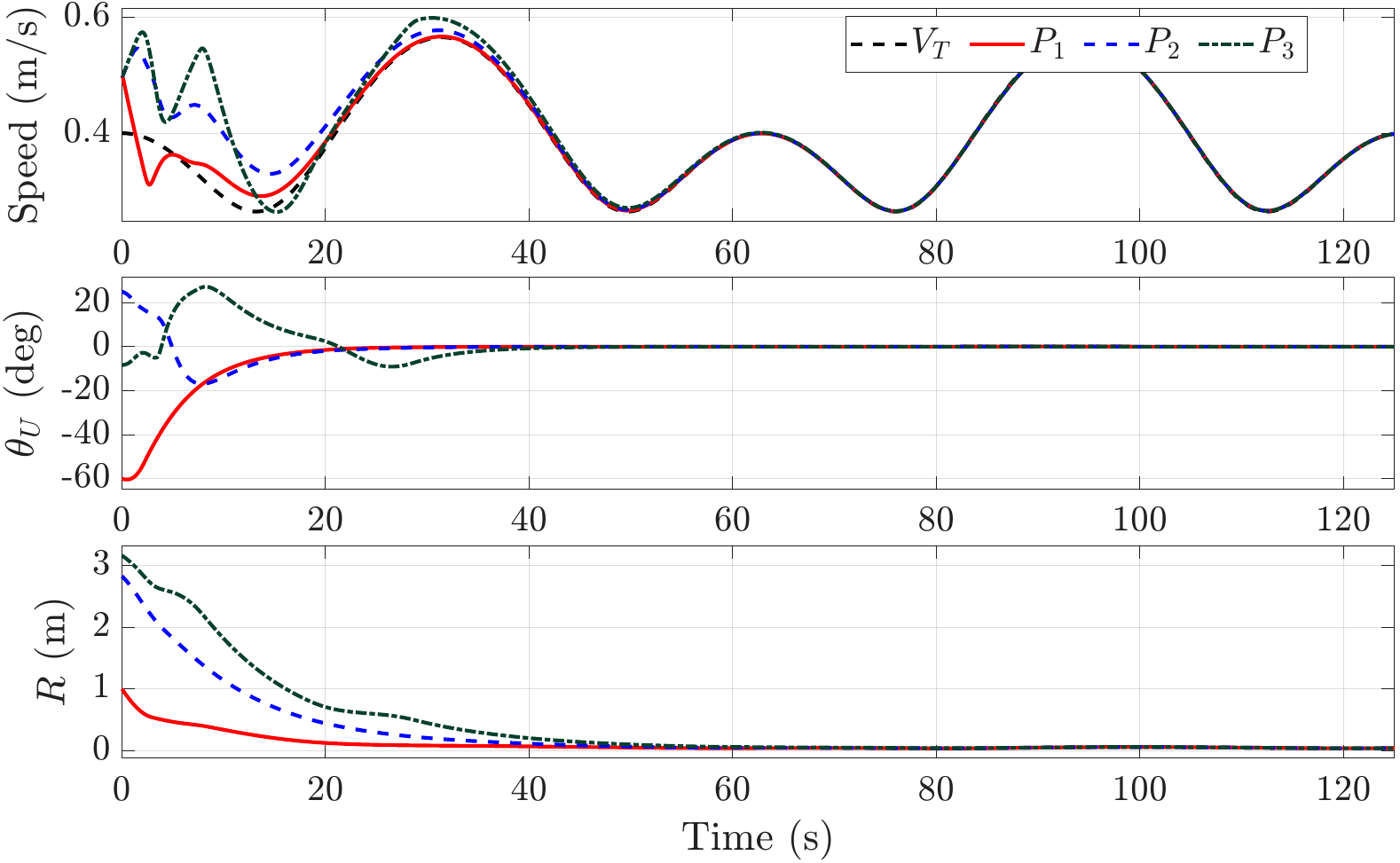}
			\caption{Speed, $\theta_U$ and $R$.}
			\label{fig:inp_sat_eight_speed_lead_R}
		\end{subfigure}
		\begin{subfigure}{0.5\textwidth}
			\includegraphics[width=\linewidth]{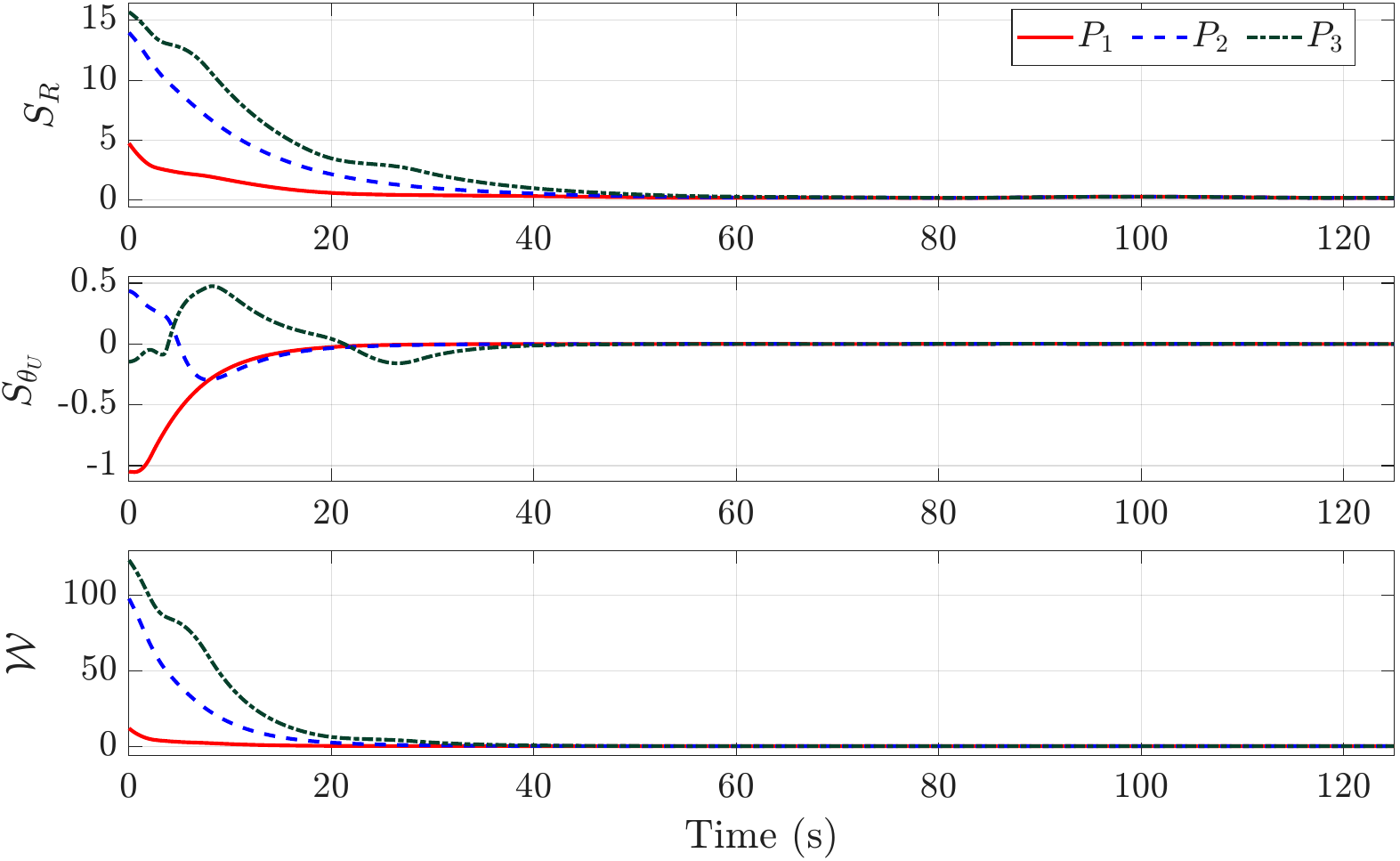}
			\caption{Evolutions of sliding surfaces and Lyapunov function.}
			\label{fig:inp_sat_eight_surface}
		\end{subfigure}%
		\begin{subfigure}{0.5\textwidth}
			\includegraphics[width=\linewidth]{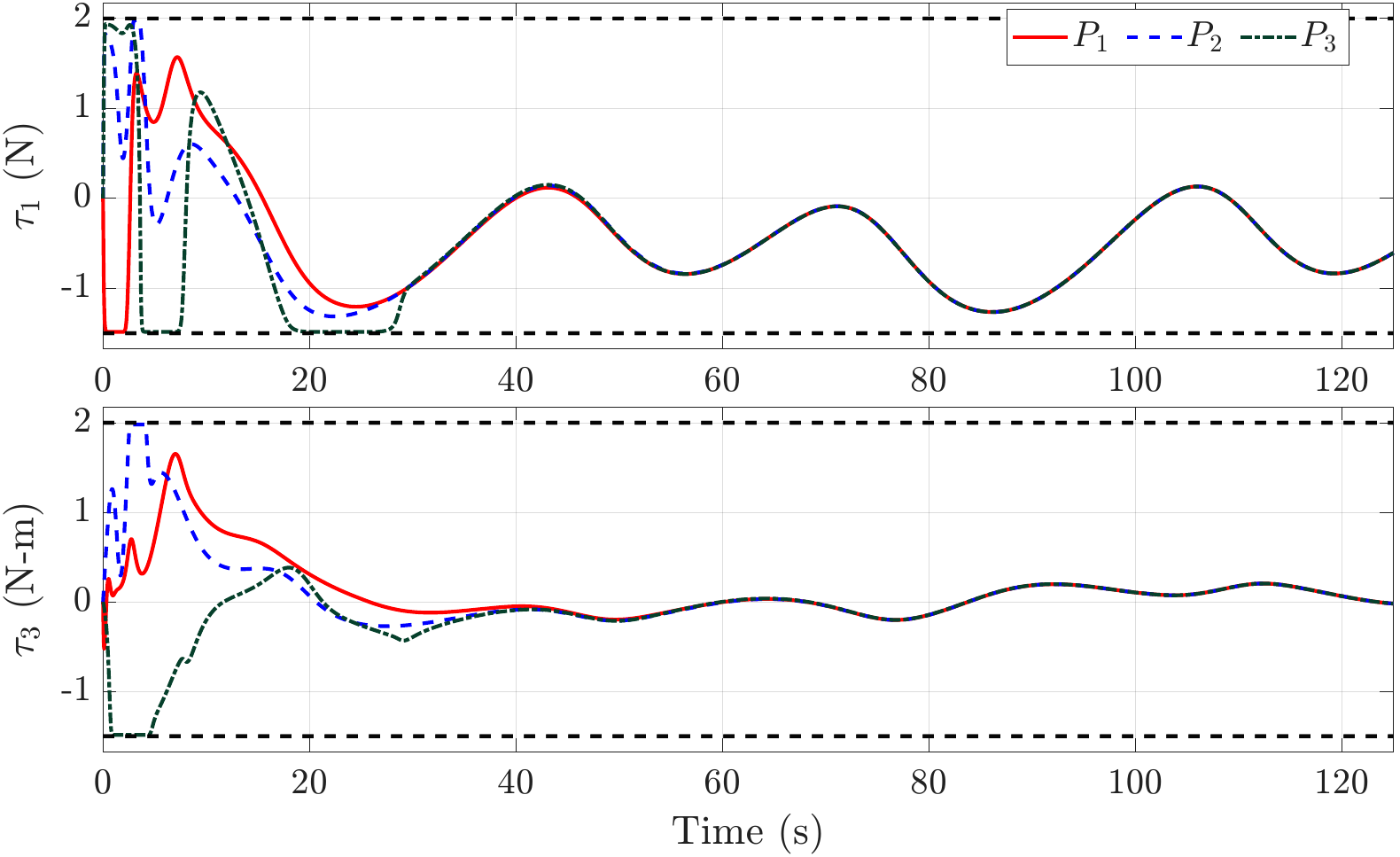}
			\caption{Control inputs.}
			\label{fig:inp_sat_eight_tau}
		\end{subfigure}
		\caption{Path-following performance for the case of an 8-shaped path under input saturation.}
		\label{fig:inp_sat_eight}
	\end{figure*}
	Three different initial conditions, namely $P_1$, $P_2$, and $P_3$, have been considered and are denoted by a triplet as $(x,\,y,\,\psi)$: $(5~\text{m},0~\text{m},120^\circ)$, $(2~\text{m},-2~\text{m},70^\circ)$, and $(5~\text{m},-3~\text{m},100^\circ)$. 
	From \Cref{fig:inp_sat_eight}, it is evident that the USV is able to follow the specified 8-shaped reference path. 
	Following the expected behavior, the USV first aligns itself with LOS (see $\theta_U$ going to zero), reaches the path (see $R$ going to zero), and catches up with the assigned speed (see \cref{fig:inp_sat_eight_speed_lead_R}). 
	The nullification of sliding surfaces and the Lyapunov function also confirms this behavior, indicating that all objectives are met and resulting in accurate path-following.
	It can be verified from \Cref{fig:inp_sat_eight_tau} that the control demands stay within the bounds, that is, $\tau_1 \in (-1.5,2)$ N and $\tau_3 \in (-1.5,2)$ N-m.
	\subsection{Comparison of results with others}
	In this section, we present the comparison of the performance of the proposed strategies. In the first IGC design strategy, the actuator bounds were ignored during the design. Later, we imposed the actuator bounds on an ad hoc basis using a saturation block. In the second IGC design strategy, bounds on the control input were explicitly accounted for in the control input derivation, ensuring that the control demand respects the actuator bounds. To perform a comparison, we consider an 8-shaped reference path with the initial configuration of the USV denoted by $P_1$ as $(5~\text{m},0~\text{m},60^\circ)$. The initial surge speed of the USV is $0.5$ m/s, with zero initial sway and yaw rates. 
	\begin{figure*}[ht!]
		\centering
		\begin{subfigure}{0.5\textwidth}
			\includegraphics[width=\linewidth]{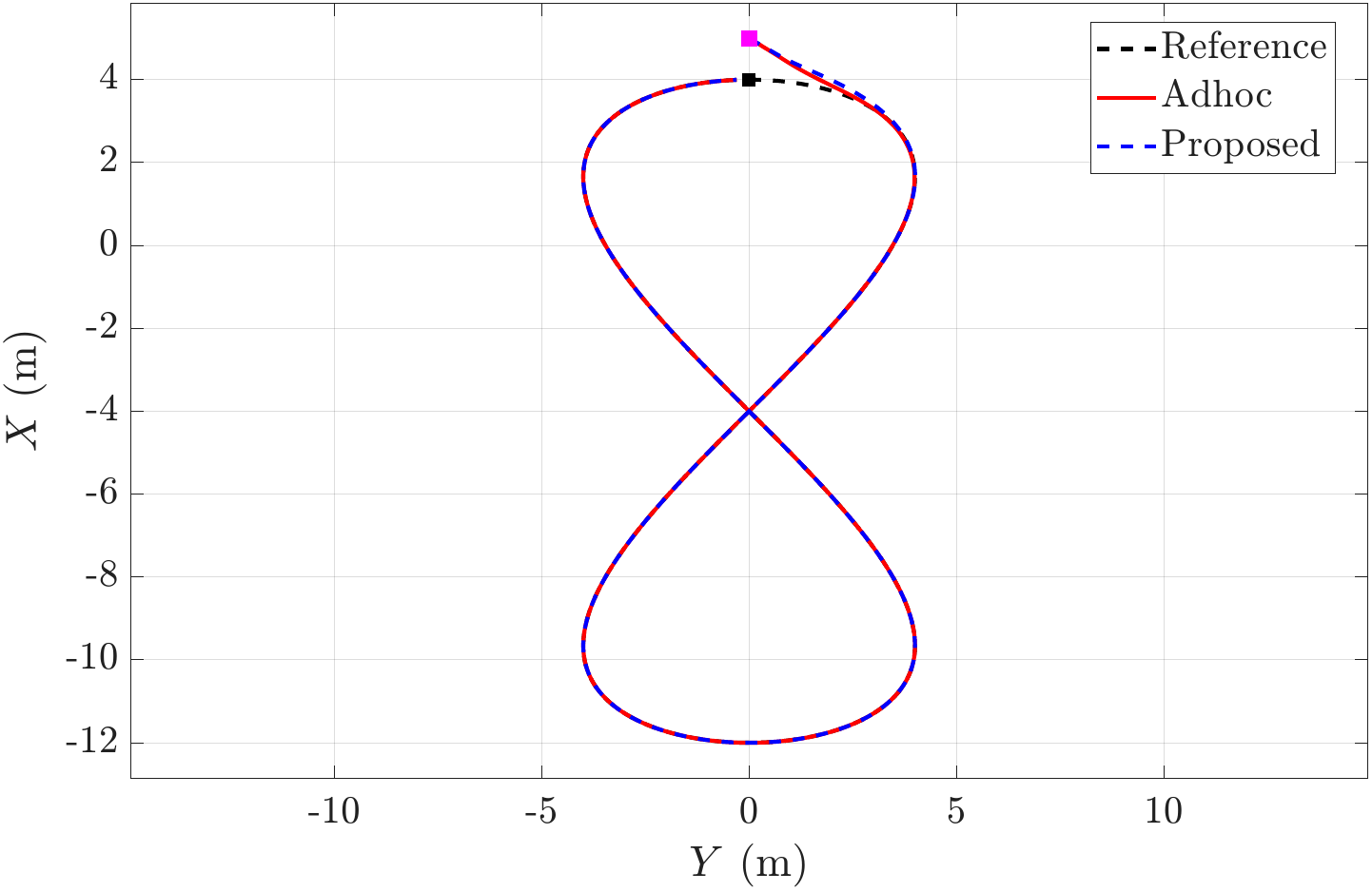}
			\caption{Path of the USV in the inertial frame.}
			\label{fig:comp_1eight_path}
		\end{subfigure}%
		\begin{subfigure}{0.5\textwidth}
			\includegraphics[width=\linewidth]{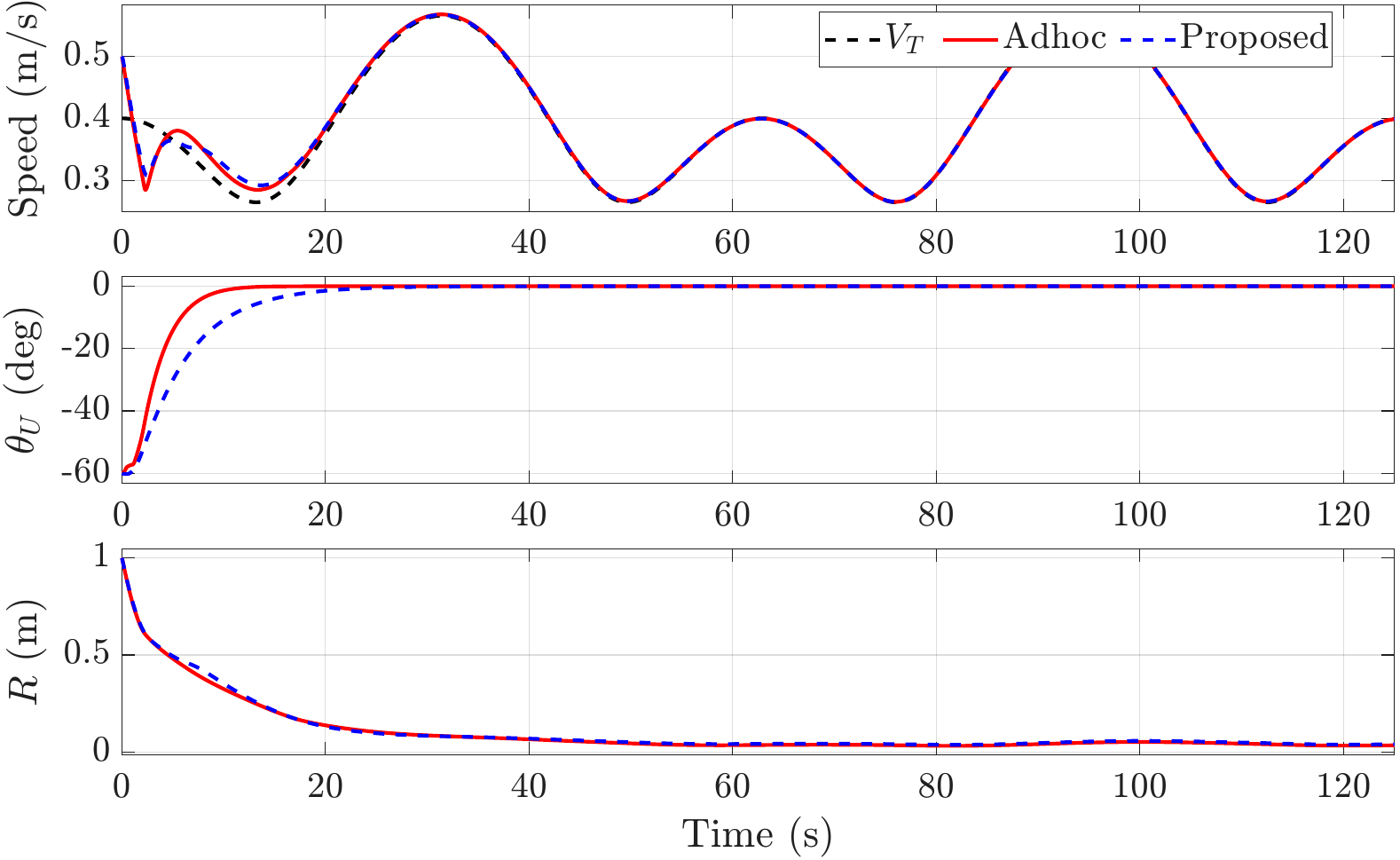}
			\caption{Speed, $\theta_U$ and $R$.}
			\label{fig:comp_1eight_speed_lead_R}
		\end{subfigure}
		\begin{subfigure}{0.5\textwidth}
			\includegraphics[width=\linewidth]{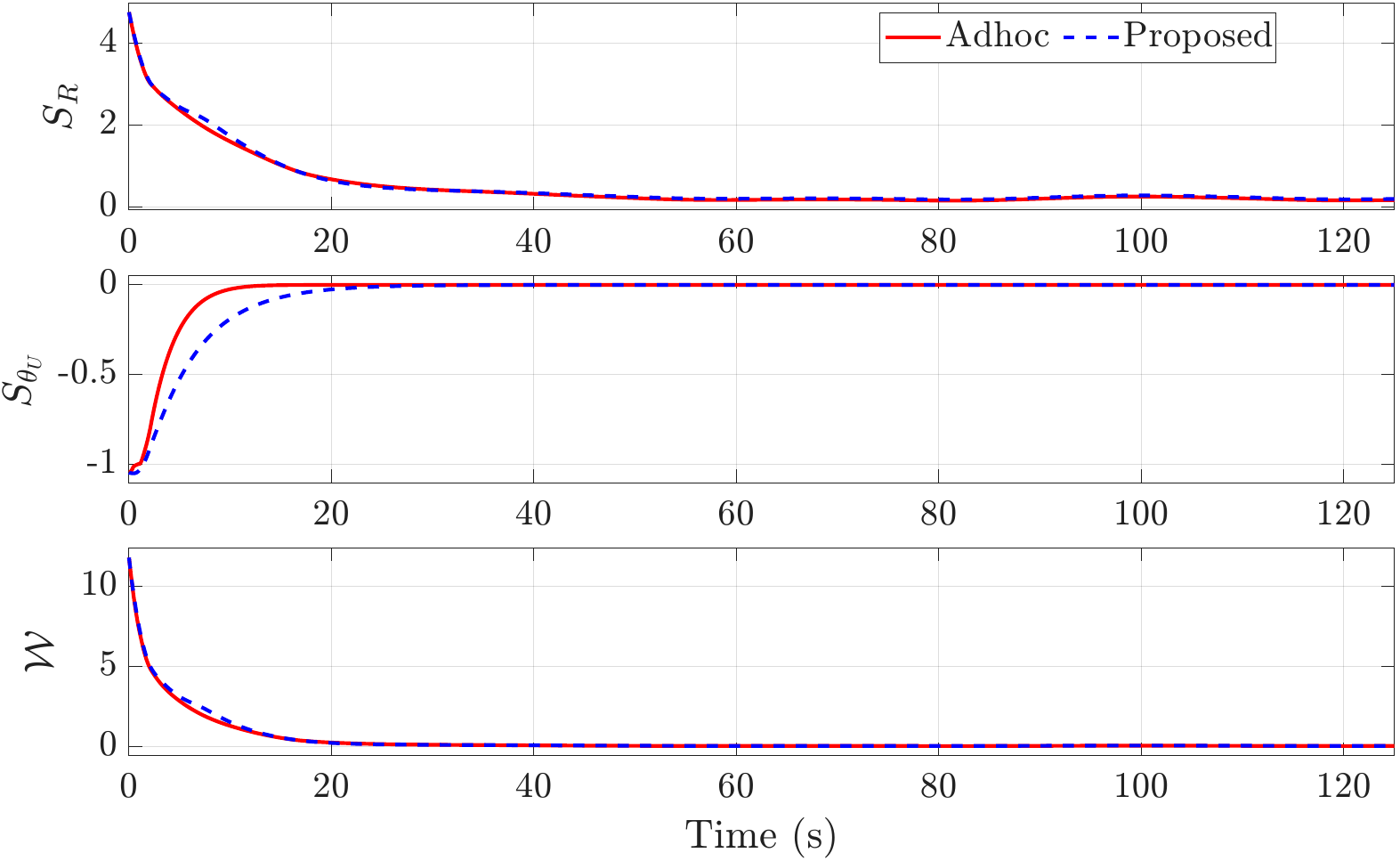}
			\caption{Evolutions of sliding surfaces and Lyapunov function.}
			\label{fig:comp_1eight_surface}
		\end{subfigure}%
		\begin{subfigure}{0.5\textwidth}
			\includegraphics[width=\linewidth]{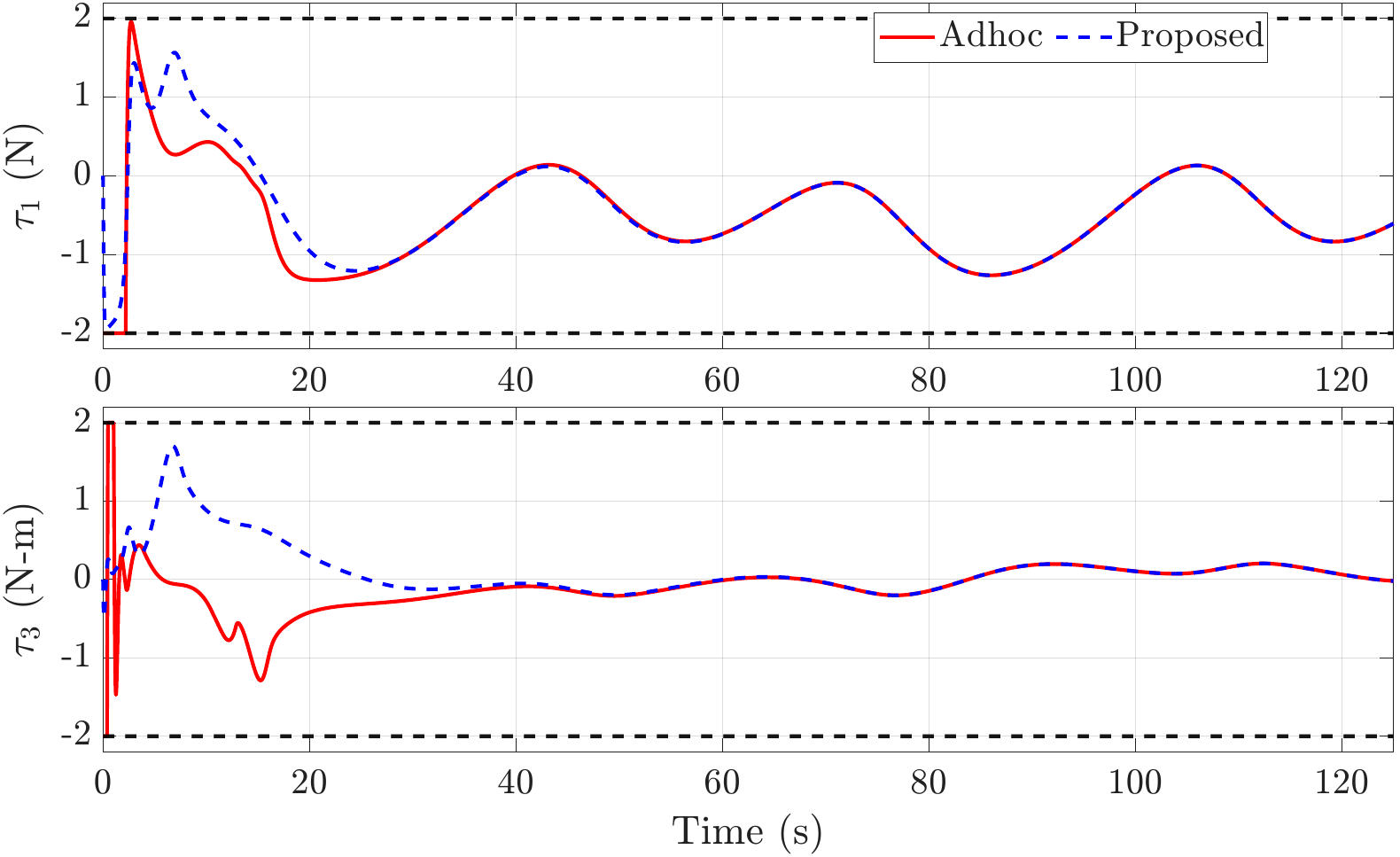}
			\caption{Control inputs.}
			\label{fig:comp_1eight_tau}
		\end{subfigure}
		\caption{Comparison of the path-following performance for an 8-shaped reference path with input saturation constraints and without input saturation constraints. }
		\label{fig:comp_1eight}
	\end{figure*}
	Considering actuator bounds in the design is important, as all practical actuators have limited control authority. 
	But in many cases, it is easier to design a stable controller assuming unlimited control authority available. 
	Hence, it is common practice to limit the control input demand to the maximum or minimum bound in cases when the demand exceeds the bounds.
	The control scheme in \Cref{eqn:tau_final} is without consideration of bounds on the inputs. 
	However, we imposed a saturation limit during simulations and labeled this `Ad hoc' in the simulation figures legend. 
	In the control scheme presented in \Cref{eqn:ip_dynamics,eqn:tauc}, the bounds on the actuator are accounted for in the design, so that specified actuator bounds are never violated. 
	The comparison among the two methods is presented in \Cref{fig:comp_1eight}. 
	With both the proposed methods, the USV can follow the specified path at varying speeds. It can be seen from \Cref{fig:comp_1eight_tau} that during the transients, the control demand from the `Ad hoc` methods is not smooth, frequently hitting upper and lower bounds, and has corners. 
	In contrast, the control input demanded by the proposed method is smooth and free of such corners. 
	This eliminates jerks in the actuators, thereby extending their lifespan by reducing wear and tear from sudden movements. 
	After the transients, the USV catches up with the virtual reference point, and the control input profile is similar. 
	This affirms that the proposed strategy does not affect the nominal performances. Above all, the second method provides a mathematical guarantee of stability under bounded input conditions. 
	\section{Conclusion}\label{sec:conclusion}
	This paper proposed nonlinear IGC strategies for a USV to follow an arbitrary smooth path with and without input saturation bounds. 
	Many existing works in literature achieve a path-following behavior along an arbitrary path by breaking the path into smaller, simpler segments, such as lines and circles, and then determining the desired heading.
	In contrast, the current work formulates this problem as an interceptor guidance problem, treating any arbitrary smooth path as a continuum of virtual target or reference points. 
	The rendezvous of the USV with this moving target results in precise path-following behavior. 
	A key point here is that, unlike interceptors, the proposed scheme allows the USV to move along the path at time-varying speeds. The stability of this nonlinear IGC strategy was formally proven using Lyapunov methods. Numerical simulations, performed using the CyberShip II model, validated the effectiveness of the proposed IGC algorithms. The results demonstrated an excellent path-following behavior across varied reference paths and USV initial configurations. The design strategy incorporating the smooth input saturation model gave a smooth control input profile, enhancing the actuators' lifespan.
	
	\bibliography{references.bib}
	
\end{document}